\documentclass{article}

\usepackage{arxiv}

\usepackage[utf8]{inputenc} 
\usepackage[T1]{fontenc}    
\usepackage{hyperref}       
\usepackage{url}            
\usepackage{booktabs}       
\usepackage{amsfonts}       
\usepackage{nicefrac}       
\usepackage{microtype}      
\usepackage{lipsum}		
\usepackage{graphicx}
\usepackage{doi}

\usepackage{graphicx}
\usepackage{amsmath}
\usepackage{amssymb}
\usepackage{amsthm}
\usepackage[ruled, vlined]{algorithm2e}
\usepackage{subcaption}
\usepackage{hyperref}
\usepackage{bm}

\newtheorem{remark}{Remark}
\newtheorem{example}{Example}
\newtheorem{lemma}{Lemma}
\newtheorem{problem}{Problem}
\newtheorem{definition}{Definition}
\newtheorem{theorem}{Theorem}
\newtheorem{corollary}{Corollary}

\title{Partitioning algorithms for weighted trees and cactus graphs}

\date{} 					

\author{ \href{https://orcid.org/0000-0002-3446-4343}{\includegraphics[scale=0.06]{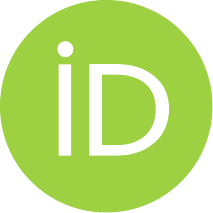}\hspace{1mm}Maike Buchin}\\
	Department of Computer Science\\
	Ruhr University Bochum\\
	\texttt{maike.buchin@rub.de} \\
	\And
	\href{https://orcid.org/0000-0002-6401-7157}{\includegraphics[scale=0.06]{orcid.pdf}\hspace{1mm}Leonie Selbach} \\
	Department of Computer Science\\
	Ruhr University Bochum\\
	\texttt{leonie.selbach@rub.de} \\
}

\renewcommand{\shorttitle}{Partitioning algorithms for weighted trees and cactus graphs}

\hypersetup{
pdftitle={Partitioning algorithms for weighted trees and cactus graphs},
pdfsubject={q-bio.NC, q-bio.QM},
pdfauthor={Maike~Buchin, Leonie~Selbach},
pdfkeywords={Graph partition, Tree, Cactus graph, Dynamic programming},
}

\begin{document}
\maketitle

\begin{abstract}
In this paper, we consider different constrained partition problems for weighted trees and cactus graphs. We focus on the $(l,u)$-partition problem, which is the problem of partitioning a weighted graph into connected clusters such that each cluster fulfills the lower and upper weight constraints $l$ and $u$. Partitioning into a minimum, maximum or a fixed number of clusters is known to be NP-hard in general, but polynomial-time solvable on trees. We prove that these three variants of the $(l,u)$-partition problem can be solved for cactus graph as well by presenting a polynomial-time algorithm. Additionally, we present an efficient method to compute the corresponding partitions. For other optimization goals or additional constraints, the partition problem becomes NP-hard - even on trees and for a lower weight bound equal to zero. We show that our method can be used as an algorithmic framework to solve other partition problems for weighted trees and cactus graphs with a pseudopolynomial runtime.
\end{abstract}

\keywords{Graph partition \and Tree \and Cactus graph \and Dynamic programming}

\section{Introduction}\label{se:intro}
Graph partitioning is an algorithmic tool that has many applications including image processing, scientific simulations and the analysis of complex networks such as social or road networks~\cite{schloegel2003graphpartitioning,Buluc2016}. Our research is motivated by an application in the field of bioinformatics, namely the fragmentation of tissue samples. As we will explain later, underlying problem of constrained shape decomposition can be reduced to the partition of weighted cactus graphs~\cite{socg20}. These are graphs in which every two simple cycles have at most one common vertex. Here, we study different constrained partition problems on cactus graphs as well as trees. A preliminary version of this paper appeared in~\cite{eurocg20}.\par 
In the following, a \emph{partition} of a graph $G=(V,E)$ refers to a partition of the vertex set $V$ into disjoint connected subsets $V_i$ such that $\bigcup V_i=V$. We call these subsets \emph{components} or \emph{clusters}. The \emph{size} of a partition $P=\{V_1,V_2,\ldots,V_p\}$ is the number of clusters and is denoted by $\vert P\vert$. A partition of size $p$ is called a \emph{$p$-partition}.  We consider partition problems on weighted graphs. Weights can be assigned to vertices or edges or both. We mostly focus on vertex-weighted graphs, where every vertex $v$ is assigned a non-negative integer weight $w(v)$. The weight of a cluster $V_i$ is defined as the sum of the weights of the vertices it contains, i.e., $w(V_i) = \sum_{v\in V_i}w(v)$. Let $P$ be a partition and $l$ and $u$ two non-negative integer parameters with $l\leq u$. $P$ is called \emph{$(l,u)$-partition} if the weight of each cluster in $P$ lies between $l$ and $u$, i.e., $l\leq w(V_i)\leq u$ for all $V_i\in P$. A $p$-$(l,u)$-partition is an $(l,u)$-partition of size $p$. A \emph{MinNum-} or \emph{MaxNum}-$(l,u)$-partition is an $(l,u)$-partition with the minimum or respectively maximum number of clusters, see for example Figure~\ref{fig: 2example}. \par 

\begin{figure}[t]
	\centering
	\begin{subfigure}[c]{0.48\textwidth}
		\includegraphics[scale=0.8]{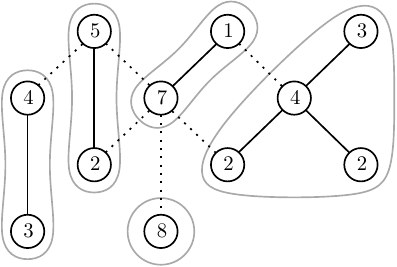}
		\subcaption{5-(3,12)-partition.}
	\end{subfigure}
	\hfill
	\begin{subfigure}[c]{0.48\textwidth}
		\includegraphics[scale=0.8]{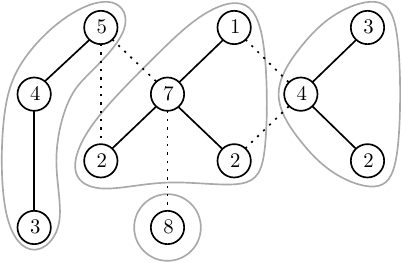}
		\subcaption{MinNum-(3,12)-partition.}
	\end{subfigure}
	\caption{Example of a $p$- and MinNum-$(l,u)$-partition of a weighted cactus graph. The values inside the vertices are their weights.}
	\label{fig: 2example}
\end{figure}

\subsection{Related work} 
Many graph partitioning problems are known to be NP-hard in general but polynomial-time solvable for certain graph classes. Next, we discuss results for related partition problems. We focus on vertex partitions of weighted graphs into connected clusters.\par 
Partitioning a graph into connected clusters such that each cluster contains a fixed number $k$ of vertices is NP-complete even for planar bipartite graphs, but can be solved for both trees and series-parallel graphs in polynomial time~\cite{DYER1985139}. This problem describes a special case of the $(l,u)$-partition problem if the graph has unit weights and $k=l=u$. Partitioning a graph such that the weights of all clusters are as equal as possible is known as the equipartition problem. This problem proved to be NP-hard for spiders and thus trees in general, but polynomial-time solvable for other graph classes, such as stars, worms and caterpillars~\cite{de1990fair}.\par 
Let $G$ be a graph with weights on the vertices and costs on the edges. Finding a partition such that the weight of the clusters is bounded by $u$ while minimizing the cost of the partition is NP-hard even on trees.  However, algorithms with pseudopolynomial runtime exist. One can solve this problem for trees in $\mathcal{O}(u^2n)$ time with a bottom-up approach~\cite{lukes1974efficient} or in $\mathcal{O}(un^2)$ time with a left-to-right approach~\cite{johnson1983knapsacks}. For sibling graphs, this problem can be solved in $\mathcal{O}(u^3n)$ time~\cite{bordawekar2008algorithm}. If we have unit weights on the edges, the problem reduces to the MinNum $(0,u)$-partition problem, namely finding a partition with constrained weights while minimizing the number of clusters. For trees, such a partition can be found in linear time~\cite{kundu1977linear}. Hamacher et al. considered a different variant of this problem, where the tree has $q$ different weights on each vertex, which are bounded by individual weight bounds $u_1,\ldots,u_q$. They proved that finding a partition with minimal size is NP-hard if $q$ or the maximum degree of the tree is unbounded~\cite{hamacher2000tree}.\par 
Perl and Snir considered a weight and capacity constrained partition problem in which not only the weight but also the capacity of the clusters is bounded from above and the goal is to find partition with the minimum number of clusters. They proved that this problem is NP-hard for trees and presented pseudopolynomial-time algorithms~\cite{perl1983circuit}. \par 
The problem of finding a $p$-partition of a vertex-weighted graph while minimizing the weight of the heaviest cluster is known as the min-max $p$-partition problem. The max-min $p$-partition problem is defined analogously. Both problems are NP-hard for general graphs~\cite{perl1981max}, but can be solved for trees in polynomial time. The first approach of Becker et al. used a shifting algorithm technique, which also works for different weight functions~\cite{becker1983shifting,becker1995shifting}. Megiddo provided polynomial-time algorithms for trees and paths~\cite{megiddo1984applying} and the runtime was further improved by Cole~\cite{Cole1987slowing}. Finally, Frederickson presented an algorithm that solves the max-min and min-max problems on trees in linear time~\cite{frederickson1991optimal}. The size-constrained min-max $p$-partition problem is another variant in which two weight functions (denoted as weights and sizes) are applied  on the vertices. This problem has the additional condition that the size of the clusters is bounded from above. Agasi et al. showed that this problem is NP-hard on trees and presented a pseudopolynomial-time algorithm, which also applies the shifting technique~\cite{agasi1993shifting}.\par 
Lucertini et al. showed that the $p$-$(l,u)$-partition problem can be solved on paths in linear time~\cite{lucertini1993most}. Ito et al. proved that this problem is NP-hard for series-parallel graphs and thus graphs in general~\cite{ITO2006142}. They presented pseudopolynomial-time algorithms for series-parallel graphs as well as partial $k$-trees.   Furthermore, Ito et al. proved that the decision variant of these problems can be solved on trees in polynomial time~\cite{Ito2012}. The respective runtimes are presented in Table~\ref{tab: runtimes}. Another variant of these problems consider a graph with $q$ different weights on the vertices and $q$ individual weight bounds $l_i$ and $u_i$. The corresponding MinNum and MaxNum as well as the $p$-partition problems can be solved on series-parallel graphs in $\mathcal{O}(u^{4q}n)$ and respectively $\mathcal{O}(p^2u^{4q}n)$ time~\cite{ito2007partitioning}.\par 
There are many problems, such as the graph coloring problem or the maximum independent set problem, that can be solved on cactus graphs in polynomial or even linear time~\cite{das2012some}. Because all cactus graphs have a tree width of at most two, they present a subclass of partial $k$-trees with $k=2$. Thus, one can use the method of Ito et al.~\cite{ITO2006142} to solve the $p$-$(l,u)$-partition problem for cactus graphs in $\mathcal{O}(p^2u^6n)$ time. The MinNum and MaxNum $(l,u)$-partition can be solved in $\mathcal{O}(u^6n)$ time. Note that both of these runtimes are pseudopolynomial because they depend on the upper weight bound $u$. In this paper, we present an algorithm that solves these problems for cactus graphs in polynomial time. 

\begin{table}[h]
	\begin{tabular}{l|l|l}
		graph class
		& $p$-partition problem  
		& MinNum/MaxNum problem \\ \hline
		path~\cite{lucertini1993most}                & $\mathcal{O}(n)$                
		& $\mathcal{O}(n)$   \\
		tree (decision~\cite{Ito2012})         &   $\mathcal{O}(p^4n)$  
		& $\mathcal{O}(n^5)$ \\
		\textbf{tree (computation}) &     $\bm{\mathcal{O}(p^4n)}$  
		& $\bm{\mathcal{O}(n^5})$ \\
		\textbf{cactus graph}         &   $\bm{\mathcal{O}(p^4n^2)}$   
		& $\bm{\mathcal{O}(n^6)}$  \\
		series-parallel graph~\cite{ITO2006142} & NP-complete, $\mathcal{O}(p^2u^{4}n)$                         
		& NP-hard, $\mathcal{O}(u^4n)$   \\
		partial $k$-tree~\cite{ITO2006142}  &  NP-complete, $\mathcal{O}(p^2u^{2(k+1)}n)$ 
		& NP-hard, $\mathcal{O}(u^{2(k+1)}n)$ 
	\end{tabular}
	\caption{Complexity of the general $(l,u)$-partition problems (computation variant unless stated otherwise) for different graph classes. The bold results are presented in this paper.}
	\label{tab: runtimes}
\end{table}

\subsection{Contribution and paper organization}
We consider different weight-constrained partition problems. Our contribution includes the following main aspects:
\begin{itemize}
	\item We consider the general $(l,u)$-partition problems on cactus graphs. We extend the partition approach for trees to cactus graphs by including a procedure that deals efficiently with cycles in the graph. We show that this method solves the $p$-$(l,u)$-partition problem as well as the MinNum and MaxNum partition problems in polynomial time. 
	\item We present a method to compute $(l,u)$-partitions for both trees and cactus graphs.
	Previous research included a polynomial-time algorithm that was restricted to the decision problem and a computation method with pseudopolynomial runtime. We show that the computation problem can be solved in polynomial time as well. 
	\item We show that our method can be used as an algorithmic framework to solve other weight-constrained partition problems. The partition method presented in this paper can be adjusted for different NP-hard partition problems to obtain solutions with a pseudopolynomial runtime. We consider each problem for both trees and cactus graphs. The resulting algorithms extend and improve on other known results.   
\end{itemize}
In Section~\ref{sec: motivation}, we present the motivation behind our research and show how a shape decomposition problem can be reduced to the considered partition problem for cactus graphs. In Section~\ref{sec: preliminiaries}, we introduce preliminary definitions and notations. In Section~\ref{sec: general l,u partition problems}, we consider the general $(l,u)$-partition problems for cactus graphs and present polynomial-time algorithms for the decision as well as computation problems. In Section~\ref{sec: other problems}, we show how the presented partition methods for trees and cactus graphs can be adjusted to obtain solutions for a selection of other weight-constrained partition problems. In Section~\ref{sec: conclusion}, we conclude with some remarks and open problems.

\section{Motivation}\label{sec: motivation}
As mentioned before, our research is motivated by a shape decomposition problem, which occurs in the microdissection of tissue samples. This problem can be formalized as computing a decomposition of a simple polygon such that every subpolygon fulfills certain constraints in size and shape. In our previous research, we developed a skeleton-based decomposition approach~\cite{socg20,selbach2021shape}. Our method creates a decomposition based on skeleton branches and allows the implementation of various additional constraints. The computation has polynomial runtime and works efficiently in practice. However, extending this approach to general decompositions leads to an exponential runtime. For certain constraints, i.e. monotone constraints such as area, we can reduce this skeleton-based decomposition problem to the partition of cactus graphs. In the following, we briefly introduce the underlying problem and present the reduction.\par
The \emph{skeleton} or medial axis of a polygon $P$ is defined as the set of points that are the centers of maximal disks inside of $P$.
The points, where the maximal disk touches the boundary of $P$, are called \emph{contact points}. Let $C(s)$ be the set of contact points of a skeleton point $s$. The \emph{degree} $deg(s)=\vert C(s)\vert$ of a skeleton point is defined as the number of contact points. The skeleton points with $deg(s)\geq 3$ are called \emph{branching points}. If the polygon $P$ is simple, the skeleton graph is a tree. We apply a skeletonization method, which creates a discrete and simplified skeleton. The computed skeleton consists of a finite number of skeleton points and represents the main morphological features of the polygon. We developed a skeleton-based decomposition approach to takes these morphological features into account.\par 
In our skeleton-based approach, we allow only cuts created by line segments between skeleton points and their contact points. We say a subpolygon $P'$ of $P$ is generated by two skeleton points $s$ and $t$ if it consists of two consecutive line segments for each of the two skeleton points and the boundary of $P$ that lies in between the corresponding contact points. If the skeleton point $t$ is an end point of a skeleton branch, the subpolygon generated by $(s,t)$ spans the boundary of $P$ between two corresponding contact points of $s$. For a polygon $P$ and its skeleton $S$, we denote the set of all possible subpolygons created by two adjacent skeleton points by $Z(P,S)$ (see Fig.~\ref{subfig: poly with skeleton}). If there are $n$ skeleton points in $S$, we have $\vert Z(P,S)\vert = n-1$.
We call two subpolygons that are generated by two pairs of skeleton points $(s,t)$ and $(s,t')$ \emph{adjacent} if they share at least one line segment at the common skeleton point $s$. We call a set $A$ of disjoint subpolygons $P_i$ a \emph{cluster} if they are adjacent in such a way that their union creates a single subpolygon $P(A)=\bigcup_{P_i\in A}P_i$. \par
\begin{figure}[t]
	\centering
	\begin{minipage}{\textwidth}
		\begin{subfigure}[t]{0.3\textwidth}
			\centering
			\includegraphics[scale=0.8]{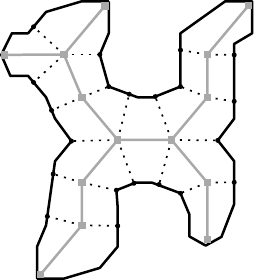}
			\caption{Polygon $P$ with its skeleton $S$ (gray) and the maximal skeleton-based partition $Z(P,S)$ (dotted)}
			\label{subfig: poly with skeleton}
		\end{subfigure}
		\hfill
		\begin{subfigure}[t]{0.3\textwidth}
			\centering
			\includegraphics[scale=0.8]{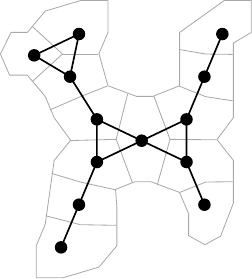}
			\caption{Cactus graph $G_P$ resulting from $P$ and $S$.}
			\label{subfig: poly with cactus}
		\end{subfigure}
		\hfill
		\begin{subfigure}[t]{0.3\textwidth}
			\centering
			\includegraphics[scale=0.8]{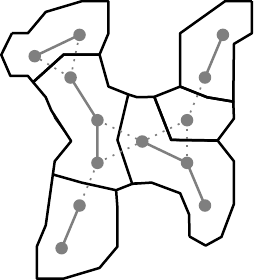}
			\caption{Partition of $G_P$ and the corresponding decomposition of $P$.}
			\label{subfig: cactus poly partition}
		\end{subfigure}
	\end{minipage}
	\caption{Reduction of the skeleton-based decomposition problem for simple polygons to a partition problem for cactus graphs. For simplicity, the skeleton is illustrated with a small sample of skeleton points.}
	\label{fig: polytocactus}
\end{figure}

We consider the following decomposition problem, which we call \emph{MinNum $(l,u)$-decomposition problem}: Given a simple polygon $P$ and size constraints $l$ and $u$ ($l\leq u$), compute a decomposition of $P$ into the minimum number of subpolygons such that the area of each subpolygon $a(P_i)$ lies between $l$ and $u$. We call a subpolygon \emph{feasible} if it fulfills $l\leq a(P_i)\leq u$. For the skeleton-based decomposition,  we can formalize the problem as finding a partition of the set $Z(P,S)$ into $p$ disjoint clusters $C_1,\ldots,C_p$ such that $p$ is minimal and each subpolygon created by a cluster $C_i$ is feasible.\par 
We can reduce the problem of finding a MinNum skeleton-based $(l,u)$-decomposition of a simple polygons to the problem of finding a MinNum $(l,u)$-partition of a weighted cactus graph (see Fig.~\ref{fig: polytocactus}). Given a simple polygon $P$, its skeleton $S$ and the set $Z(P,S)$, we create a weighted graph $G_P=(V,E,w)$ as follows: We have a vertex $i\in V$ for every subpolygon $P_i\in Z(P,S)$. The weight of each vertex $w(i)=a(P_i)$ equals the area of the subpolygon and we have an edge $(i,j)\in E$ if $P_i$ and $P_j$ are adjacent. 
Let $Q_s=\{P_i\vert s\in P_i\}$ be the set of all subpolygons containing the skeleton point $s$. We have $deg(s)=\vert Q_s\vert$. 
Note that the vertices $V(Q_s)=\{i\in V\vert P_i\in Q_s\}$ form a simple cycle in $G_P$ if and only if $deg(s)\geq 3$. As each subpolygon $P_i$ is generated by two skeleton points, each $P_i$ is contained in two sets $Q_s$ and $Q_t$. Thus, the vertex $i$ is contained in at most two cycles and $G_P$ is a cactus graph. Note that if the skeleton of the polygon $P$ has $n$ skeleton points, the cactus graph $G_P$ has $n-1$ vertices. A MinNum $(l,u)$-partition of $G_P$ consists of clusters $C_i$ of vertices such that the subpolygons created by the clusters are feasible. Thus, the graph partition can be transformed into a feasible solution for the MinNum skeleton-based $(l,u)$-decomposition problem. With our result from Section~\ref{sec: minimum and maximum partition}, we are able to obtain a feasible solution in polynomial runtime. Nevertheless, the runtime is quite expensive and therefore other methods, such as our branch-wise approach, might still be favorable in the practical application.

\section{Preliminaries}\label{sec: preliminiaries}
Let $G=(V,E)$ be a graph with $n$ vertices and $m$ edges. A \emph{tree} is a graph without any cycles. Given node $v$ in a rooted tree $T$, we denote the subtree rooted in $v$ by $T_v$. For the root node $r$, we have $T_r=T$. We define $T_v^i$ as the subtree induced by $v$ and its first $i$ children (see Fig.~\ref{fig:subtree}).
\begin{figure}[b]
	\centering
	\includegraphics[scale=0.95]{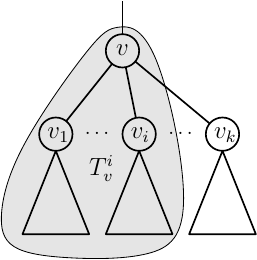}
	\caption{Subtree $T_v^i$ of a node $v$.}
	\label{fig:subtree}
\end{figure}
A \emph{cactus graph} is a graph in which two cycles have at most one vertex in common. Equivalently, these are the graphs in which every edge belongs to at most one cycle. The vertex set $V$ of a cactus graph can be partitioned into three subsets. A vertex is a \emph{C-vertex} if it has degree 2 and is included in exactly one cycle. It is a \emph{G-vertex} if it is not included in any cycle. All remaining vertices are \emph{H-vertices} and are also referred to as \emph{hinges}. Every H-vertex belongs to at least one cycle. 
Let $C= \langle w_1,w_2,\ldots,w_m \rangle$ be a cycle in the cactus graph. We define the \emph{tree representation} $T_G$ of a cactus graph $G$ as the graph such that:
\begin{itemize}
	\item all G- and H-vertices are represented with corresponding nodes and edges
	\item every cycle $C= \langle w_1,w_2,\ldots,w_m \rangle$ in $G$ is represented with a single cycle node $c$. 
	\item if a H-vertex $v$ belongs to a cycle $C$, there is an edge between the hinge node $v$ and the cycle node $c$. 
\end{itemize}
By replacing all the cycles in $G$ with cycle nodes, we obtain a tree structure for $T_G$.
An example of this is shown in Figure~\ref{fig:cactus tree representation}.
Let $C=\langle w_1,w_2,\ldots,w_m\rangle $ be a cycle that is represented by a node~$c$. For every $w_i\in C$ that is a H-vertex, there exists an edge to a hinge node corresponding to $w_i$, which we denote by $H_{w_i}$.
A \emph{rooted cactus graph} is a cactus graph in which one hinge vertex is selected as the root. In this case, $T_G$ is given as a rooted tree. In the following, when considering cactus graphs or their tree representations, we always refer to the rooted variant. Note that we refer to the elements in the cactus graph as vertices and to the elements of its tree representation as nodes.\par
\begin{figure}[t]
	\centering
	\begin{subfigure}[c]{0.50\textwidth}
		\includegraphics[scale=0.95]{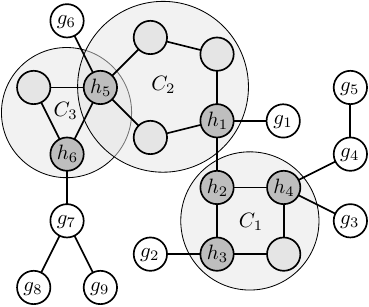}
	\end{subfigure}\hfill
	\begin{subfigure}[c]{0.47\textwidth}
		\includegraphics[scale=0.95]{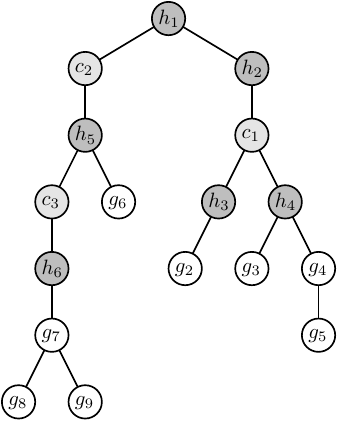}
	\end{subfigure}
	\caption{A cactus graph $G$ and its corresponding tree representation $T_G$. The H- and G-vertices are denoted with $h_i$ and $g_i$. The hinge $h_1$ is considered as the root of $G$ and therefore also the root of $T_G$. Every cycle $C_i$ is represented by some node $c_i$ in $T_G$.}
	\label{fig:cactus tree representation}
\end{figure}

Recall that a $p$-$(l,u)$-partition of a graph is a partition of the vertex set into exactly $p$ clusters such that the weight of each cluster lies between the two weight bounds $l$ and $u$. 
Let $P$ be a $(l,u)$-partition of a tree $T$. Given some node $v$, we denote the cluster that contains $v$ by $P_v$. If $l$ equals zero, $P$ induces a feasible $(0,u)$-partition for any subtree $T_v^i$. This is not necessarily true for a $(l,u)$-partition in general, as the weight of the cluster $P_v$ might be less than $l$. We call a partition of a subtree $T_v^i$ an \emph{extendable} $(l,u)$-partition if every cluster except for $P_v$ fulfills both weight constraint, meaning, $w(P_v)\leq u$ and $l\leq w(V_i)\leq u$ for all $V_i\neq P_v$. Such a partition can potentially be extended to become feasible by adding more nodes to the cluster $P_v$.\par 
Our partition algorithms operate on \emph{partition sets}. A partition set $S(v,i)$ is defined for a certain subtree $T_v^i$ in the given tree (representation). 
\begin{align}
\begin{split}
S(v,i) = \{(x,k)\ \vert\ &\exists\text{ extendable $(l,u)$-partition $P$ of $T_v^i$} \\
&\text{s.t. $\vert P\vert = k$ and $w(P_v)=x$}\}.
\end{split}
\end{align}
Every tuple $(x,k)\in S(v,i)$ corresponds to an extendable $(l,u)$-partition of the subtree $T_v^i$, which has size $k$ and in which the weight of the cluster containing the root node $v$ equals $x$. Let $c_v$ be the number of children of some node $v$. The partition set $S(v,c_v)$ is in short denoted by $S(v)$. 
\begin{remark}\label{rem: existence partition}
	A tree $T_v^i$ has a $p$-$(l,u)$-partition if and only if there exists a tuple $(x,p)\in S(v,i)$ such that $l\leq x\leq u$. 
\end{remark}

\section{General \texorpdfstring{$(l,u)$}{(l,u)}-partition problems}\label{sec: general l,u partition problems}

In this section, we consider the three general $(l,u)$-partition problems, which are finding a $(l,u)$-partition with a minimum, maximum or fixed number of clusters. 
\begin{problem}[Decision $p$-$(l,u)$-partition problem]
	Let $G=(V,E,w)$ be a graph with weights $w$ on its vertices. Given two non-negative integers $l$ and $u$ with $l\leq u$ and a positive integer $p\leq n$, is there a vertex partition into $p$ clusters $V_1, V_2, \ldots, V_p$ such that $l\leq w(V_i)\leq u$ for all $1\leq i\leq p$? 
\end{problem}
In the following, we assume that $w(v) \leq u$ for all $v\in V$. Otherwise, a $p$-$(l,u)$-partition would not exist. Ito et al.~\cite{Ito2012} presented a polynomial-time algorithm that solves the given problem if $G$ is a tree. Here, we show that a similar method can be used to solve the problem on cactus graphs as well. The general idea is to utilize the tree representation of a cactus graph and include an additional procedure that deals with cycles. 
First, we present a simple algorithm and then show how this algorithm can be adjusted to solve the given problem in polynomial time. In general, we consider only graphs with non-negative integer weights. However, the polynomial-time algorithm allows us to include real-valued weights as well.

\subsection{Simple algorithm}\label{sec: simple algorithm}
\subsubsection{Tree partition}
First, we present the algorithmic approach if a tree $T$ is given as the input. This corresponds to the method of Ito et al. adapted to our notation. The basic idea is to compute partitions for all nodes using a bottom-up and left-to-right approach. More specifically, we compute partitions for each subtree $T_v^i$ for increasing values of $i$ and for $v$ going from the leaves to the root. 
Let $v$ be a node in a tree $T$ and $v_i$ the $i$-th child of $v$. Note that we obtain a partition of $T_v^i$ by combining partitions of the subtrees $T_v^{i-1}$ and $T_{v_i}$. Given a partition $P'$ of $T_v^{i-1}$ and $P''$ of $T_{v_i}$, we can combine them to create different partitions $P$ of $T_v^i$ as follows (see Fig.~\ref{fig: partition possibilities}):
\begin{description}
	\item[(A)] Merge the two clusters containing $v$ ($P'_v \in P')$ and $v_i$ ($P''_{v_i}\in P'')$.
	\item[(B)] Join partitions without merging. 
\end{description}
In case (A), the resulting partition $P$ has size $\vert P'\vert + \vert P''\vert -1$ and the new cluster $P_v\in P$ has weight $w(P_v)=w(P'_v)+w(P''_{v_i})$. In case (B), $P$ has size $\vert P'\vert + \vert P''\vert$ and the new cluster $P_v\in P$ is equal to $P'_v$.
We can obtain all possible partitions of $T_v^i$ by combining all partitions of its subtrees in this manner.\par

\begin{figure}[t]
	\centering
	\begin{minipage}{0.8\textwidth}
		\begin{subfigure}[c]{\textwidth}
			\centering
			\includegraphics[scale=0.95]{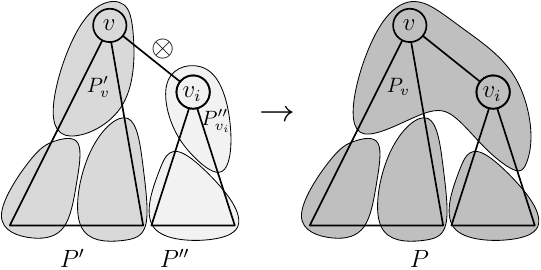}
		\end{subfigure}\\
		\begin{subfigure}[c]{\textwidth}
			\centering
			\includegraphics[scale=0.95]{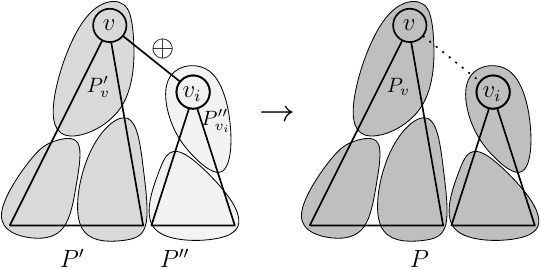}
		\end{subfigure}
	\end{minipage}
	\caption{Different possibilities to combine partitions $P'$ of $T_v^{i-1}$ and $P''$ of $T_{v_i}$ to obtain a partition $P$ of $T_v^i$. The top one corresponds to option (A) and the bottom one to option (B).}
	\label{fig: partition possibilities}
\end{figure}

For the given partition problem, we want to compute extendable $(l,u)$-partitions with at most $p$ clusters. When combining partitions, we eliminate all partitions whose size exceeds $p$. In case (A), we have to check that $P''_{v_i}$ fulfills the lower weight constraint because $P''$ has to be a feasible $(l,u)$-partition and not just an \emph{extendable} one. In case (B), we have to check that the weight of the new cluster $P_v$ does not exceed the upper weight bound $u$.\par 
This computation can be formalized using the partition sets defined in Section~\ref{sec: preliminiaries}.  We define two operations $\oplus$ and $\otimes$, which correspond to the cases (A) and (B). For two partition sets $S_1$ and $S_2$ the operations are defined as follows:
\begin{align}\label{eq: oplus otimes}
S_1 \oplus S_2 &= \{(x_1,k_1+k_2)\ \vert\ l\leq x_2, k_1+k_2\leq p, (x_1,k_1)\in S_1, (x_2,k_2)\in S_2\}\nonumber\\
S_1 \otimes S_2 &=\{(x_1+x_2,k_1+k_2-1)\ \vert\  x_1+x_2\leq u, k_1+k_2-1\leq p,\\
&\qquad (x_1,k_1)\in S_1, (x_2,k_2)\in S_2\}.\nonumber
\end{align}
The complete partition set of all feasible partitions created from $S_1$ and $S_2$ is given as the union of the results of both operations. We denote this operation with $\odot$:
\begin{equation*}\label{eq: odot}
S_1 \odot S_2 = (S_1 \oplus S_2) \cup (S_1 \otimes S_2).
\end{equation*}
Using these notations, we can compute partition sets for each $T_v^i$ using the following dynamic programming approach:
\begin{align*}
S(v,0) &= \{(w(v),1)\}\\
S(v,i) &= S(v,i-1)\odot S(v_i).
\end{align*}
These partition sets are computed for all nodes in the tree from the leaves to the root. If there exists a correct tuple $(x,p)$ in the partition set of the root node $S(r)$, we know if there exists a $p$-$(l,u)$-partition of the tree (see Remark~\ref{rem: existence partition}). Algorithm~\ref{alg: tree partition} solves the decision variant of the $p$-$(l,u)$-partition problem in this way in $\mathcal{O}(u^2p^2)$ time.

\begin{algorithm}[tb]
	\DontPrintSemicolon
	\KwIn{Tree $T$ with root $r$, integers $l,u,p$ with $0\leq l\leq u$ and $0<p$.}
	\KwOut{\texttt{yes} if there is a $p$-$(l,u)$-partition of $T$ and \texttt{no} otherwise.}
	\ForAll{$v\in V$ bottom-up}{
		$S(v,0) = \{(w(v),1)\}$\;
		\For{$1\leq i\leq c_v$}{
			$S(v,i) = S(v,i-1)\odot S(v_i)$}}
	\lIf{$(x,k)\in S(r)$ such that $l\leq x\leq u$ and $k=p$} 
	{\Return \texttt{yes}}\lElse{\Return \texttt{no}}
	\caption{Simple algorithm to decide the $p$-$(l,u)$-partition problem for a tree}
	\label{alg: tree partition}
\end{algorithm}

\subsubsection{Cactus graph partition} 
With some additional procedures, we can generalize the tree partition algorithm to solve the $p$-$(l,u)$-partition problem on cactus graphs (see Alg.~\ref{alg: cactus partition}). Let $G$ be a weighted cactus graph and $T$ its tree representation, which is rooted in some hinge $r$. We compute partition sets $S(v,i)$ for all nodes $v$ in $T$ corresponding to G- and H-vertices. $S(v,i)$ contains tuples belonging to partitions of the subgraph of $G$ that is represented with the subtree $T_v^i$. Again, the algorithm considers all nodes in a bottom-up and left-to-right manner. If $v$ and its $i$-th child $v_i$ are not cycle nodes, we can compute $S(v,i)=S(v,i-1)\odot S(v_i)$ as before. If $v_i$ is a cycle node, we execute an additional procedure \ref{alg: cyclePartition}($v_i,i$), which returns the partition set $S(v,i)$.

\begin{algorithm}[tb]
	\DontPrintSemicolon
	\KwIn{cactus graph $G$ rooted in hinge $r$, integers $l,u,p$ with $0\leq l\leq u$ and $0<p$.}
	\KwOut{\texttt{yes} if there is a $p$-$(l,u)$-partition of $G$ and \texttt{no} otherwise.}
	$T=(V',E')$ tree representation of $G$ rooted in $r$,\;
	$C'$ set of all cycle nodes in $V'$\;
	\ForAll{$v\in V'\setminus C'$ bottom-up}{
		$S(v,0) = \{(w(v),1)\}$\;
		\For{$1\leq i\leq c_v$}{
			\uIf{$v_i\in C'$}{
				$S(v,i) = \ref{alg: cyclePartition}(v_i,i)$\;}
			\Else{
				$S(v,i) = S(v,i-1)\odot S(v_i)$}}}
	\lIf{$(x,k)\in S(r)$ such that $l\leq x\leq u$ and $k=p$} 
	{\Return \texttt{yes}}\lElse{\Return \texttt{no}}
	\caption{Simple algorithm to decide the $p$-$(l,u)$-partition problem for a cactus graph}
	\label{alg: cactus partition}
\end{algorithm}

\paragraph{Cycle partition}
Let $C=\langle w_1,w_2,\ldots,w_m \rangle $ be a cycle in the cactus graph $G$ that is represented by some cycle node~$c$ in its tree representation $T$.
In general, one way of finding all partitions of a cycle is to decompose the cycle into different paths and then compute partition for these paths. A cycle of length $m$ can be decomposed into $m$ different paths by removing one edge each time. In our case, we consider these paths as small trees that are rooted in $w_1$ (see Fig.~\ref{fig: cycle partitions}). We call these the different \emph{configurations} of a cycle and number them from $1$ to $m$. In configuration $j$ the left branch of $w_1$ consists of the vertices $w_2$ to $w_j$ (empty if $j=1$) and the right branch of $w_m$ to $w_{j+1}$ (empty if $j=m$). Thus, we can find all partitions such that $w_j$ and $w_{j+1}$ (with $w_{m+1}=w_1$) are not in the same cluster - unless this cluster contains all nodes of the cycle. We can show that only $m-1$ configurations are needed to obtain all partitions of a cycle. Thus, we only consider the first $m-1$ configurations from now on.\par 
\begin{figure}[t]
	\centering
	\includegraphics[scale=1]{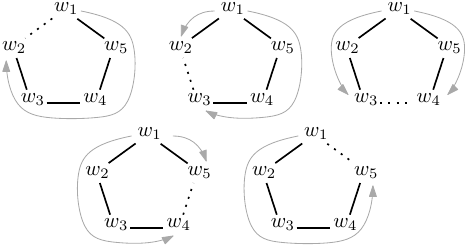}
	\caption{Configurations with respect to a root $w_1$ in a cycle of length five.}
	\label{fig: cycle partitions}
\end{figure}
\begin{lemma}
	Let $C= \langle w_1,w_2,\ldots,w_m \rangle $ be a cycle in a cactus graph $G$. 
	Given a partition $P$ of $C$, this partition can be found in one of the $m-1$ configurations of $C$. 
\end{lemma}
\begin{proof}
	We consider the configurations of $C$ with respect to a root vertex $w_1$. 
	Let $P_{w_1}$ be the cluster in the partition $P$ that contains $w_1$. Assume that $P_{w_1}=\{w_1\}\cup P_L \cup P_R$ with $P_L=\{w_2,\ldots,w_i\}$ and $P_R=\{w_j, \ldots, w_m\}$ and $1\leq i <j\leq m+1$. If $i=1$ or $j=m+1$, then $P_L$ or respectively $P_R$ is considered to be empty. We show that for each possible values of $i$ and $j$ exists a configuration unequal to $m$ in which this cluster can be found.\\
	If $j=i+1$, $P_{w_1}$ contains all nodes of the cycle and the partition can be found in any configuration. If $j\neq i+1$, the partition contains additional clusters and can be found if all remaining vertices ($\notin P_{w_1}$) are contained in the same (either left or right) subtree of $w_1$. This is the case in two configurations, namely, $i$ and $j-1$.
	If we consider configuration $j-1=m$, then there exists another configuration $i$ in which the partition can be found. Because we assumed $j\neq i+1$, we have $i<j-1=m$. Thus, this configuration is unequal $m$.
	 \end{proof} 
\bigskip

\begin{figure}
	\centering
	\begin{subfigure}[b]{0.45\textwidth}
		\centering
		\includegraphics[scale=1]{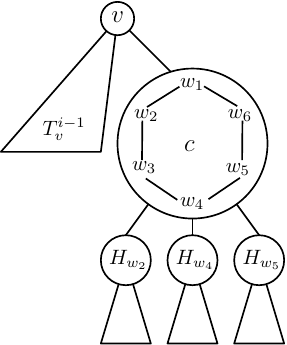}
		\subcaption{}
	\end{subfigure}\hspace{1cm}
	\begin{subfigure}[b]{0.45\textwidth}
		\centering
		\includegraphics[scale=1]{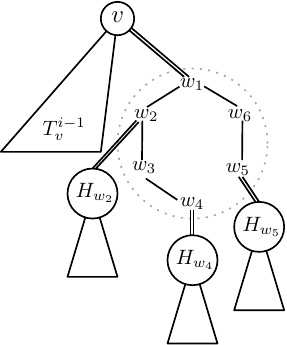}
		\subcaption{}
	\end{subfigure}
	\caption{(a): Example of subtree $T_v^i$ of a hinge $v=H_{w_1}$ with the $i$-th child being a cycle node $c$ corresponding to a cycle $C= \langle w_1,w_2,\ldots, w_6 \rangle $ in the cactus graph. (b): Subgraph considered in configuration 4 of the cycle. The subgraph in the dotted circle is $T(C,4)$. The double edges indicate that the two connected vertices/nodes are the same.}
	\label{fig: tree config}
\end{figure}

To compute partitions of the cycle $C=\langle w_1,w_2,\ldots,w_m \rangle $ represented by a cycle node $c$, we utilize this idea of considering different configurations. 
Let $v_c$ be a node in $T$ with the cycle node $c$ as its $i_c$-th child. When computing the partition set $S(v_c,i_c)$ the procedure \ref{alg: cyclePartition}($c,i_c$) is executed. For every H-vertex $v$ in $C$, we denote its corresponding hinge node in $T$ 
by $H_v$. Let $w_1$ be the vertex such that $H_{w_1}=v_c=pred(c)$ i.e. the parent of $c$ in $T$. For all other H-vertices~$v$, $H_v$ is a child of $c$.\par

In configuration $j$ of the cycle $C$, the edge $(w_j,w_{j+1})$ is deleted. We consider the remaining graph as a tree rooted in $w_1$, which we denote by $T(C,j)$ (see Figure~\ref{fig: tree config}). For each configuration~$j$, we compute individual partition sets $S_j$. Let $b_v$ be the number of children of the vertex $v$ in $T(C,j)$. In configuration $j$, we have:
\begin{align*}
b_{v} = \begin{cases}
0 & \mbox{if $v=w_j$ and $j>1$ or $v=w_{j+1}$, }\\
1 & \mbox{if $v=w_1$ and $j=1$},\\
2 & \mbox{if $v=w_1$ and $j\neq 1$},\\
1 & \mbox{otherwise.}
\end{cases}
\end{align*}
For a vertex $v$ and $0<i\leq b_v$, we compute $S_j(v,i)$ as in the tree partition algorithm:
\begin{align*}
S_j(v,i) = S_j(v,i-1)\odot S_j(v_i).
\end{align*}
However, the partition sets $S_j(v,0)$ have to be initialized differently to take previously computed partitions into account. 
The initialization is as follows:
\begin{align}
S_j(v,0) = \begin{cases}
\{(w(v),0)\} & \mbox{if $v$ C-vertex},\\
S(H_{v}) & \mbox{if $v\neq w_1$ H-vertex},\\
S(H_{v},i_c-1) & \mbox{if $v=w_1$ H-vertex}.
\end{cases}
\end{align}
After the partition sets for all configurations are computed, we compute a final partition set $S_c$ as the union of the partition sets $S_j(w_1)$.
\begin{equation*}
S_c= \bigcup_{j=1}^{m-1} S_j(w_1).    
\end{equation*}
This partition set is returned by the function and matches $S(pred(c),i_c)=S(v_c,i_c)$.

\begin{function}[tb]
	\DontPrintSemicolon
	\KwIn{Cycle node $c$ representing the cycle $C=\langle w_1,w_2,\ldots, w_{m}\rangle$, integer $i_c$}
	\For{$j=1$ \KwTo $m-1$}{
		Let $T(C,j)$ be the tree rooted in $w_1$ corresponding to the cycle $C$ in configuration $j$ and $b_v$ the number of children of $v$ in $T(C,j)$\;
		\ForAll{$v$ bottom-up}{
			\uIf{$v=w_1$}{
				$S_j(v,0) = S(pred(c),i_c-1)$\;
			}\uElseIf{$w_i$ H-vertex corresponding to a child node $H_{w_1}$}{
				$S_j(v,0) = S(H_{v})$\;
			}\Else{
				$S_j(v,0) = \{(w(v),1)\}$\;
			}
			\For{$1\leq i\leq b_v$}{
				$S_j(v,i) = S_j(v,i-1)\odot S_j(v_i)$
			}
		}
	}
	$S_c=\bigcup_{j=1}^{m-1} S_j(w_1)$\;
	\Return $S_c$
	\caption{CyclePartition($c$,$i_c$)}
	\label{alg: cyclePartition}
\end{function}	




\paragraph{Runtime and correctness}
Let $G$ be a cactus graph with tree representation $T$. We denote the part of $G$ that is represented by some subtree $T_v^i$ by $G(T_v^i)$. Thus, $G=G(T)$. Let $c$ be a cycle node that is contained in the subtree $T_v^i$ and let $C$ be the cycle in $G$ that is represented by $c$. If we consider $C$ in its configuration $j$, we denote the corresponding subgraph by $G(T_v^i,c,j)$. Depending on the configuration, the number of children for the vertices inside the circle changes. For a vertex $v'\in C$, let $c_{v'}$ be the number of children of $v'$ in $G(T_v^i,c,j)$. We have $c_{v'}=a_{v'}+b_{v'}$ with $b_{v'}$ being the number of children contained in $C$ and $a_{v'}$ the number of children of $v$, which are not contained in $C$. Note that it is always the first $a_{v'}$ children of $v'$ that are not contained in $C$, and we have $b_{v'}\in\{0,1,2\}$. We denote the subgraph of $G(T_v^i,c,j)$ induced by $v'$ and its first $i'$ children by $G_{v'}^{i'}(T_v^i,c,j)$.
\begin{lemma}\label{lem: partition iff tuple}
	Let $G$ be a cactus graph with tree representation $T$. For every partition set computed in Algorithm~\ref{alg: cactus partition} the following holds:
	\begin{itemize}
		\item A partition set $S(v,i)$ computed in Algorithm~\ref{alg: cactus partition} contains a tuple $(x,k)$ if and only if there exists an extendable $(l,u)$-partition $P$ of $G(T_v^i)$ with $w(P_v)=x$ and $\vert P\vert=k\leq p$.
		\item A partition set $S_j(v,i)$ computed in the execution of $S(v_c,i_c)=$ \ref{alg: cyclePartition}$(c,i_c)$ in Algorithm~\ref{alg: cactus partition} contains a tuple $(x,k)$ if and only if there exists an extendable $(l,u)$-partition $P$ of $G_{v}^{a_{v}+i}(T_{v_c}^{i_c},c,j)$ with $w(P_v)=x$ and $\vert P\vert=k\leq p$.
	\end{itemize}
\end{lemma}
\begin{proof}
	We show the statement of this lemma inductively. When considering a partition set $S(v,i)$ or $S_j(v,i)$, our hypothesis is that statement is true for every partition set that was computed with Algorithm~\ref{alg: cactus partition} previous to the considered set. We start with the base case $S(v,0)$. $G(T_v^0)$ consists only of the vertex $v$, therefore we have $w(P_v)=w(v)$ and $\vert P\vert=1$. The partition set $S(v,0)$ contains only the tuple $(w(v),1)$. Thus, the statement is true. In the following, we show that the statement is true for all partition set that are computed in later stages of the algorithm. The proof is structured into the three different computations that take place in the algorithm.
	\begin{description}
		\item[1)]$S(v,i) = S(v,i-1)\odot S(v_i)$.
	\end{description}
	Consider a partition set $S(v,i)$ that is computed with the formula $S(v,i-1)\odot S(v_i)$ where $v_i$ is the $i$-th child of $v$ in $T$. For the first direction, let $P$ be an extendable $(l,u)$-partition of $G(T_v^i)$ with $w(P_v)=x$ and $\vert P\vert = k\leq p$. Let $P'$ and $P''$ be the restriction of $P$ to $G(T_v^{i-1})$ and $G(T_{v_i})$ respectively. We consider two cases:\\
	Case 1: $v_i \notin P_v$. $P'$ is an extendable $(l,u)$-partition of $G(T_v^{i-1})$ with $w(P'_v)=w(P_v)=x$ and $\vert P'\vert =k'<k$. $P''$ is a $(l,u)$-partition of $G(T_{v_i})$ with $\vert P''\vert = k-k'$ and $w(P''_v)=x''$ with $l\leq x''\leq u$. With our induction hypothesis, we have $(x,k')\in S(v_i-1)$ and $(x'',k-k')\in S(v_i)$. When these tuples are combined with the $\oplus$-operation, we obtain $(x,k'+k-k')=(x,k)\in S(v,i)$.\\
	Case 2: $v_i \in P_v$. $P'$ and $P''$ are both extendable $(l,u)$-partitions with $\vert P'\vert = k'<k$, $\vert P''\vert=k-k'+1$, $w(P'_v)=x'<x$ and $w(P''_{v_i})=x-x'$. With our induction hypothesis, we have $(x',k')\in S(v,i-1)$ and $(x-x',k-k'+1)\in S(v_i)$. These tuples are combined with the $\otimes$-operation and we obtain $(x'+x-x',k'+k-k+1-1)=(x,k)\in S(v,i)$.\\
	When considering the other direction, let $(x,k)$ be a tuple contained in $S(v,i)$. We know that $(x,k)$ was computed by combining two tuples $(x_1,k_1)\in S(v,i-1)$ and $(x_2,k_2)\in S(v_i)$. With our induction hypothesis these tuples correspond to extendable $(l,u)$-partitions $P'$ of $G(T_v^{i-1})$ and $P''$ of $G(T_{v_i})$ respectively. Again, we consider two cases:\\
	Case 1: $(x,k)$ was computed with the $\oplus$-operation. Therefore, we have $x=x_1\leq u$ and $k=k_1+k_2\leq p$ and $l\leq x_2$. We can combine all clusters of $P'$ and $P''$ to create a partition $P=P'\cup P''$ and have $w(P_v)=w(P'_v)=x$ and $\vert P\vert=k_1+k_2=k$. Because the cluster $P_{v_i}$ fulfills the lower weight constraint $l$, $P$ is an extendable $(l,u)$-partition.\\
	Case 2: $(x,k)$ was computed with the $\otimes$-operation. Therefore, we have $x=x_1+x_2\leq u$ and $k=k_1+k_2-1\leq p$. We can combine $P'$ and $P''$ as follows: $P=(P'\setminus P'_v)\cup (P''\setminus P''_{v_i})\cup \{P'_v\cup P''_{v_i}\}$. Then, we have $w(P_v)=w(P'_v)+w(P''_{v_i})=x_1+x_2=x$ and $\vert P\vert=\vert P'\vert -1 + \vert P''\vert -1 +1 =k_1+k_2+1=k$. $P$ is obviously a $(l,u)$-partition.
	\begin{description}
		\item[2)]$S_j(v,i)$.
	\end{description}
	Consider a partition set $S_j(v,i)$ that is computed when executing $S(v_c,i_c)=\ref{alg: cyclePartition}(c,i_c)$. 
	First, we consider the case that $i=0$. We have the following property:
	\begin{align*}
	G_v^{a_v+0}(T_{v_c}^{i_c},c,j) = \begin{cases}
	(\{v\},\emptyset) &\mbox{ if $v$ is C-vertex (a)}\\
	G(T_{v_c}^{i_c-1}) &\mbox{ if $v$ H-vertex with $H_v = v_c$ (b)}\\
	G(T_{H_v}) &\mbox{ otherwise (c)}
	\end{cases}
	\end{align*}
	Case (a) represents the base case we considered above. There is only one partition $P$ of the graph consisting only of the vertex $v$ and we have $w(P_v)=w(v)$ and $\vert P\vert=1$. Since we have $S_j(v,0)=\{(w(v),1)\}$, the statement is true. For the other two cases, we again consider both directions.
	For the first direction, let $P$ be an extendable $(l,u)$-partition of $G_v^{a_v+i}(T_{v_c}^{i_c},c,j)$ with $w(P_v)=x$ and $\vert P\vert = k\leq p$. In case (b), $P$ is a partition of $G(T_{H_{v_c}}^{i_c-1})$. With our induction hypothesis, we have a corresponding tuple $(x,k)\in S(v_c,i_c-1)=S_j(v,0)$. In case (c), $P$ is a partition of $G(T_{H_{v}})$. With our induction hypothesis, we have a corresponding tuple $(x,k)\in S(H_v)=S_j(v,0)$.\\
	For the other direction, let $(x,k)$ be a tuple in $S_j(v,0)$. Therefore, we have the same tuple $(x,k)$ contained in $S(v_c,i_c-1)$ or $S(H_v)$ respectively. Because of our induction hypothesis, the tuple corresponds to an extendable $(l,u)$-partition of the respective graph $G_v^{a_v}(T_{v_c}^{i_c},c,j)$.\\
	If we consider the case $i\neq 0$, we can prove the statement with the same reasoning presented in the paragraph 1) of this proof. Thus, the statement is true for all computed partition sets $S_j(v,i)$.
	\begin{description}
		\item[3)] $S(v,i)=\ref{alg: cyclePartition}(v_i,i)$.
	\end{description}
	Consider a partition set $S(v,i)$ that is computed as $S(v,i)=$ \ref{alg: cyclePartition}$(v_i,i)$ with $v_i=c$ being a cycle node. For the first direction, let $P$ be an extendable $(l,u)$-partition of $G(T_v^i)$ with $w(P_v)=x$ and $\vert P\vert = k\leq p$. Let $C=(w_1,w_2,\ldots,w_m)$ be the cycle represented by the node $c$ such that $H_{w_1}=v$. Assume, there are two vertices $w_j$ and $w_{j+1}$ of $C$ that are not contained in the same cluster in $P$. Then, $P$ is also a partition of the graph $G(T_v^i,c,j)$. Note that $G(T_v^i,c,j)$ equals $G_{w_1}^{c_{w_1}}(T_v^i,c,j)$ for $c_{w_1}$ being the number of children of $w_1$ in $G(T_v^i,c,j)$. We have $c_{w_1}=a_{w_1}+b_{w_1}$ for some value $i'$. With our induction hypothesis, we have a tuple $(x,k)$ corresponding to $P$, which is contained in the partition set $S_j(w_1)=S_j(w_1,b_{w_1})$. Since we have $S(v,i)=\bigcup_{j=1}^{m-1}S_j(w_1)$, the tuple $(x,k)$ is also contained in $S(v,i)$. If all vertices of the cycle are contained in one cluster in $P$, then deleting one edge from the cycle does not influence the connectivity and the partition $P$ is also a partition of $G(T_v^i,c,j)$ for any of the $m-1$ configurations of $C$. Similarly, we have $(x,k)\in S(v,i)$.\\
	For the other direction, let $(x,k)$ be a tuple contained in $S(v,i)$. Since $S(v,i)=\bigcup_{j=1}^{m-1}S_j(w_1)$, there exists one set $S_j(w_1)=S_j(w_1,b_{w_1})$ in which $(x,k)$ is also contained. With our induction hypothesis follows that there exists an extendable $(l,u)$-partition $P$ of $G_{w_1}^{a_{w_1}+b_{w_1}}(T_v^i,c,j)=G(T_v^i,c,j)$. This partition is obviously also a partition of $G(T_v^i)$. Thus, the statement of the Lemma is true for all computed partition sets. 
	 \end{proof}
\medskip
\begin{theorem}
	Let $G$ be a weighted cactus graph and $T$ its tree representation. Given a positive integer $p$ and two non-negative integers $l$ and $u$ with $l\leq u$, Algorithm~\ref{alg: cactus partition} solves the decision $p$-$(l,u)$-partition problem in time $\mathcal{O}(u^2p^2n^2)$.
\end{theorem}
\begin{proof}
	The correctness of the algorithm follows from Lemma~\ref{lem: partition iff tuple}. 
	Let $r$ be the root of the tree representation $T$. Lemma~\ref{lem: partition iff tuple} implies that there exists a $p$-$(l,u)$-partition of the cactus graph $G=G(T_r)$ if and only if there exists a tuple $(x,p)\in S(r)$ such that $l\leq x\leq u$.\\
	The runtime results from the following observations. Every partition set that is computed with the algorithm consists of tuples $(x,k)$ with $x\leq u$ and $k\leq p$. Thus, these sets are of size $\mathcal{O}(up)$. Since the $\oplus$- and $\otimes$-operations combine all elements of two partition sets with each other, one $\odot$-operation takes $\mathcal{O}(u^2p^2)$ time. 
	The $\odot$-operation is always applied to compute a partition set for some node $v$ and some child node $v_i$. We say the $\odot$-operation is performed on the edge $(v,v_i)$. Note that we only perform $\odot$-operations on edges $(v,v_i)$ in the edge set of the cactus graph $G$. If both $v$ and $v_i$ are not C-vertices, there is an edge $(v,v_i)$ in the tree representation of $G$ and we perform exactly one $\odot$-operation for this edge. If $v$ and/or $v_i$ is a C-vertex, we perform at most one $\odot$-operation on the edge $(v,v_i)$ for every configuration $j$ of the cycle during the execution of the \ref{alg: cyclePartition} procedure. Since we consider $m-1$ configurations for a cycle of length $m$, we perform $\mathcal{O}(m)=\mathcal{O}(n)$ $\odot$-operations for each edge. Since there are $\mathcal{O}(n)$ edges in a cactus graph, there are $\mathcal{O}(n^2)$ $\odot$-operations. Moreover, it takes $\mathcal{O}(up)$ time to check if there is a correct tuple in $S(r)$ for $r$ being the root of the tree representation and thus decide if a $p$-$(l,u)$-partition exists. This results in an overall runtime of $\mathcal{O}(u^2p^2n^2)$.\par 
	
	 \end{proof}

\subsection{Polynomial-time algorithm}\label{sec: poly time algorithm}
The runtimes of Algorithms \ref{alg: tree partition} and \ref{alg: cactus partition} are pseudopolynomial because they depend on $u$, which is the upper weight bound. The runtime is dominated by the time it takes to perform $\odot$-operations, which in turn depends on the size of the partition sets. Until now, the computed partition sets are of size $\mathcal{O}(up)$, because we store tuples $(x,k)$, where $x$ ranges from 1 to $u$ and $k$ from 1 to $p$. Similarly to Ito et al., we can use intervals to reduce the size of these sets and thereby obtain a polynomial-time algorithm. Specifically, we reduce the size of the partition sets to $\mathcal{O}(p^2)$ and thus decrease the runtime of Algorithm \ref{alg: cactus partition} to $\mathcal{O}(p^4n^2)$. As this reduction follows the idea of Ito et al., we only present the general outline and show that it can be extended from trees to cactus graphs. We refer to their work for missing proofs.\par 
The general idea is to store the computed weights in the tuples not as individual weights $x$ but as intervals of weights $[a,b]$. Each interval is the \emph{interval of a maximal d-consecutive subset} of weights. 
\begin{definition}
	Let $A$ be an ordered set of integers. $A$ is called \emph{$d$-consecutive} if the difference between any two consecutive elements in $A$ is at most $d$. That is, for each $a\in A\setminus \{\max(A)\}$ there exists $a'\in A$ such that $0\leq a'-a\leq d$.
\end{definition}
\begin{definition}
	Given a set $A$ and a subset $A'\subset A$, we call $A'$ a \emph{maximal} $d$-consecutive subset of $A$ if $A'$ is $d$-consecutive and there is no other $d$-consecutive subset in $A$ that contains $A'$.
\end{definition}
\begin{definition}
	Given an ordered set $A$, we call $\left[\min(A),\max(A)\right]$ the \emph{interval} of $A$.
\end{definition}
We denote the set containing all intervals of maximal $d$-consecutive subsets of a set $A$ as $I(A)$ (see Fig.~\ref{fig: interval}):
\begin{equation*}
I(A)=\{[a,a']\ \vert \ [a,a'] \text{ is the interval of a maximal $d$-consecutive subset of $A$}\}.
\end{equation*}
Let $I$ be a set of $d$-consecutive intervals. To obtain a set that only contains maximal  $d$-consecutive intervals, we have to merge all $d$-interfering intervals in $I$.
\begin{definition}
	Let $\left[a,a'\right]$ and $\left[b,b'\right]$ be two intervals such that $a\leq b$. We say $\left[ a,a'\right]$ is \emph{d-interfering} with $\left[b,b'\right]$ if $b-a'\leq d$.
\end{definition}
For two $d$-interfering intervals, we define a merge-operation $\uplus$ as follows:
\begin{equation*}
\left[a,a'\right]\uplus \left[b,b'\right] = \left[a,\max{a',b'}\right].
\end{equation*}
If both intervals are $d$-consecutive, the merged interval is $d$-consecutive as well. Given a set $I$ of intervals, we define $M(I)$ as the set of intervals we obtain by repeatedly merging all $d$-interfering intervals in $I$ until none remain (see Fig.~\ref{fig: merged interval}). Note that the set $M(I)$ is unique, does not depend on the merging order and can be computed with a runtime that is linear in the size of set $I$.\par 

\begin{figure}
	\centering
	\includegraphics[scale=0.85]{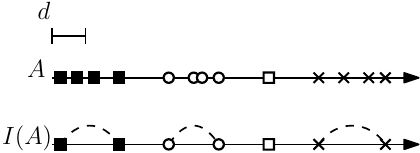}
	\caption{Interval set $I(A)$ consisting of all maximal $d$-consecutive subsets of an integer set~$A$.}
	\label{fig: interval}
\end{figure}

\begin{figure}
	\centering
	\includegraphics[scale=0.85]{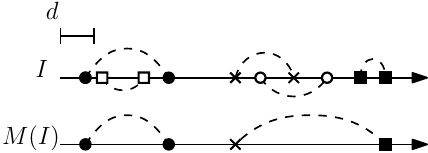}
	\caption{Merged interval set $M(I)$ obtained by merging all $d$-interfering intervals in the interval set $I$.}
	\label{fig: merged interval}
\end{figure}

The partition sets $S(v,i)$ computed in Algorithm \ref{alg: cactus partition} contain tuples $(x,k)$ corresponding to partitions of the subgraph $G(T_v^i)$ that is represented by the subtree $T_v^i$ in the tree representation of~$G$. We define
\begin{equation*}
S(v,i,k) = \{x\ \vert \ (x,k) \in S(v,i)\}.    
\end{equation*}
We substitute the tuples $(x,k)$ with tuples of the form $(\left[a,b\right],k)$ such that $\left[a,b\right]$ is the interval of a maximal $d$-consecutive subset of $S(v,i,k)$, i.e. the elements of $I(S(v,i,k))$.  The value $d$ is set as the difference between the given upper and lower weight bound, thus, $d=u-l$. We define \emph{interval partition sets} $I(S(v,i))$ as follows:
\begin{equation}\label{eq: interval partition set}
I(S(v,i)) = \{(\left[a,b\right],k)\ \vert \ \left[a,b\right] \in I(S(v,i,k))\}.
\end{equation}
For the partition sets $S_j$, which are used in \ref{alg: cyclePartition} and correspond to the subgraph created by a certain configuration $j$ of a cycle in $G$, this is defined analogously. 
\begin{lemma}
	For each $(x,k)\in S(v,i)$, we have $([a,a'],k)\in I(v,i)$ such that $a\leq x\leq a'$.
\end{lemma}
\begin{lemma}
	For each $([a,a'],k)\in I(v,i)$, we have both $(a,k)$ and $(a',k)$ in $S(v,i)$ and the set $\{x\in S(v,i,k)\ \vert \ a\leq x\leq a'\}$ is $d$-consecutive.
\end{lemma}
The algorithm is adjusted to use the interval partition sets to determine if a given graph has a $p$-$(l,u)$-partition. Note that the following lemma stays true if $T$ is the tree representation of a cactus graph.
\begin{lemma}\label{lem: poly existence partition}
	A tree $T$ with root $r$ has a $p$-$(l,u)$-partition if and only if the set $I(S(r))$ contains a tuple $(\left[a,a'\right],p)$ such that $\left[a,a'\right]\cap \left[l,u\right]\neq \emptyset$.
\end{lemma}
\begin{proof}
	First, we assume that $T$ has a $p$-$(l,u)$-partition. From Remark \ref{rem: existence partition} follows that there is a tuple $(x,p)\in S(r)$ such that $l\leq x\leq u$, i.e. $x\in \left[ l,u\right]$. The weight $x$ is contained in some interval $\left[a,a'\right]\in I(S(r,p))$ with $a\leq x\leq a'$. Therefore, we have the tuple $(\left[a,a'\right],p)\in S(r)$ and $\left[a,a'\right]\cap \left[l,u\right]\neq \emptyset$.\\
	Now, assume that there is a tuple $(\left[a,a'\right],p)\in S(r)$ such that $\left[a,a'\right]\cap \left[l,u\right]\neq \emptyset$. We show that there exists a tuple $(x,p)\in S(r)$ with $l\leq x\leq u$. We consider two cases.
	\begin{description}
		\item[Case 1:] $a'\leq u$.\\
		The tuple $(a',p)$ is contained in $S(r)$. Since $\left[a,a'\right]\cap \left[l,u\right]$, we have $l\leq a'\leq u$.
		\item[Case 2:] $a\leq u \leq a'$.\\
		Let $X={x_1,x_2,\ldots,x_k}$ be the maximal $d$-consecutive subset of $S(r,p)$ such that $a=x_1 < x_2 < \ldots < x_k=a'$. Let $i$ be the smallest integer such that $u\leq x_i$. If $u= x_i$, we have $(x_1,p)\in S(r)$ with $l\leq x_i\leq u$. If $u<x_i$, we know that $i>1$ and there exists an element $x_{i-1}$ such that $x_{i-1}<u$ and $(x_{i-1},p)\in S(r)$. Using the fact that $X$ is $d$-consecutive, we can show that $l\leq x_{i-1}$.
		\begin{align*}
		x_i - x_{i-1} &\leq d = u-l \quad \Rightarrow\quad l \leq \underbrace{u - x_i}_{<0} + x_{i-1}. 
		\end{align*}
	\end{description}
	 \end{proof}

Now, we present the adjustments needed for the algorithm to compute interval partition sets. We denote the computed sets with $I_{(j)}(v,i)$. Later we show that these computed sets are indeed equal to the interval partition sets $I_{(j)}(S(v,i))$ we defined above.
When computing interval partition sets, we use the same $\odot$-operation as before, but redefine the $\oplus$- and $\otimes$-operations from Equation~\ref{eq: oplus otimes} to be compatible with the intervals. 
\begin{align}\label{eq: oplus otimes2}
I_1 \oplus I_2 = &\{([a,a'],k_1+k_2)\ \vert\ [b,b']\cap [l,u]\neq \emptyset, k_1+k_2\leq p,\nonumber\\
& ([a,a'],k_1)\in I_1, ([b,b'],k_2)\in I_2\}\nonumber\\
I_1 \otimes I_2 = &\{([a+b, a'+b'],k_1+k_2-1)\ \vert\ a+b\leq u, k_1+k_2-1\leq p,\\
& ([a,a'],k_1)\in I_1, ([b,b'],k_2)\in I_2\}.\nonumber\\
I_1\odot I_2 = & (I_1\oplus I_2)\cup(I_1\otimes I_2)\nonumber
\end{align}
Then we can compute interval partition sets $I'$ as follows:
\begin{align*}
I(v,0) &= \{([w(v),w(v)],1)\} \\
I'(v,i) &= I(v,i-1)\odot I(v_i)
\end{align*}
Note that if $[a,a']$ and $[b,b']$  are $d$-consecutive the resulting intervals in Equation \ref{eq: oplus otimes2} are $d$-consecutive as well. Thus, all intervals in $I'(v,i)$ are $d$-consecutive. However, the interval from a computed tuple $([a,a'],k)\in I'(v,i)$ might be $d$-interfering with another interval from a tuple $([b,b'],k)$ in this set. Therefore, we use the merge-operation as defined above to obtain the set $I(v,i)$, which only contains intervals of \emph{maximal} $d$-consecutive subsets (for each $k$), as follows:
\begin{equation*}
I(v,i) = M(I'(v,i)) = \{([a,a'],k)\ \vert \ [a,a']\in M(I'(v,i,k))\}
\end{equation*}
When partitioning a cycle, we use the same approach (see \ref{alg: cyclePartition poly}). We can replace the sets $S$ with $I$ and perform additional merge-operations.
\begin{align*}
I'_j(v,i) &= I_j(v,i-1)\odot I_j(v_i)\\
I_j(v,i) &= M(I'_j(v,i)).
\end{align*}
In the end, all interval partition sets are united and merged again and the resulting set is returned by the function. 
\begin{align*}
I'_c &= \bigcup_{j=1}^{m-1} I_j(w_1)\\
I_c &= M(I'_c).
\end{align*}
The following lemma from Ito et al. shows that interval partition sets $I(v,i)$ for a subtree $T_v^i$ are indeed equal to the interval partition sets $I(S(v,i))$ defined in Equation~\ref{eq: interval partition set}. The lemma still holds if the interval partition set corresponds to a subgraph that is represented by a subtree of the tree representation or in case of $I_j(v,i)$ corresponds to a subgraph represented by a subtree with a certain cycle configuration. 
\begin{lemma}
	The set $I(v,i)$ is equal to $I(S(v,i))$.
\end{lemma}
Lemma~\ref{lem: number intervals bounded} shows that instead of having at most $u$ different weights $x$ for each $k$ ($1\leq k\leq p$), we have at most $k=\mathcal{O}(p)$ intervals $\left[a,b\right]$. Thus, we compute sets $I_{(j)}$ of size $\mathcal{O}(p^2)$ instead of $\mathcal{O}(up)$.
\begin{lemma}\label{lem: number intervals bounded}
	The number of elements in $I(S(v,i,k))$ does not exceed $k$.
\end{lemma} 
\begin{corollary}\label{cor: number interval sets}
	The size of $I(v,i)$ is in $\mathcal{O}(p^2)$.
\end{corollary}
\begin{algorithm}[h]
	\DontPrintSemicolon
	\KwIn{cactus graph $G$ rooted in hinge $r$, integers $l,u,p$ with $0\leq l\leq u$ and $0<p$.}
	\KwOut{\texttt{yes} if there is a $p$-$(l,u)$-partition of $G$ and \texttt{no} otherwise.}
	$T=(V',E')$ tree representation of $G$ rooted in $r$,\;
	$C'$ set of all cycle nodes in $V'$\;
	\ForAll{$v\in V'\setminus C'$ bottom-up}{
		$I(v,0) = \{([w(v),w(v)],1)\}$\;
		\For{$1\leq i\leq c_v$}{
			\uIf{$v_i\in C'$}{
				$I(v,i) = \ref{alg: cyclePartition poly}(v_i)$\;}
			\Else{
				$I'(v,i)=I(v,i-1)\odot I(v_i)$\;
				$I(v,i) = M(I'(v,i))$\;}}}
	\lIf{$(x,k)\in I(r)$ such that $l\leq x\leq u$ and $k=p$} 
	{\Return \texttt{yes}}\lElse{\Return \texttt{no}}
	\caption{Polynomial-time algorithm to decide the $p$-$(l,u)$-partition problem for a cactus graph}
	\label{alg: cactus partition poly}
\end{algorithm}
\begin{function}[H]
	\DontPrintSemicolon
	\KwIn{Cycle node $c$ corresponding to a cycle $C=\langle w_1,w_2,\ldots, w_{m}\rangle $, integer $i_c$}
	\For{$j=1$ \KwTo $m-1$}{
		Let $T(C,j)$ be the tree rooted in $w_1$ corresponding to the cycle $C$ in configuration $j$ and $b_v$ the number of children of $v$ in $T(C,j)$.\\
		\ForAll{$v$ bottom-up}{
			\uIf{$w_i=w_1$}{
				$I(w_1,0) = I(pred(c),i_c-1)$\;
			}
			\uElseIf{$w_i$ H-vertex corresponding to a child node $H_{w_1}$}{
				$I(w_i,0) = I(H_{w_i})$\; 
			}\Else{ $I(w_i,0) = \{(w(v_i),1)\}$}
		}
		\For{$i=1$ \KwTo $b_v$}{
			$I'_j(v,i) = I_j(v,i-1)\odot I_j(v_{i})$\;
			$I_j(v,i) = M(I'_j(v,i))$\;
	}}
	$I'_c= \bigcup_{j=1}^{m-1} I_j(w_1)$\;
	$I_c = M(I'_c)$\;
	\Return $I_c$	
	\caption{PolyCyclePartition($c$,$i_c$)}
	\label{alg: cyclePartition poly}
\end{function}
\begin{theorem}
	Let $G$ be a weighted cactus graph and $T_G$ its tree representation. Given a positive integer $p$ and two non-negative integers $l$ and $u$ with $l\leq u$, Algorithm~\ref{alg: cactus partition poly} solves the decision $p$-$(l,u)$-partition problem in time $\mathcal{O}(p^4n^2)$.
\end{theorem}
\begin{proof}
	The correctness of Algorithm \ref{alg: cactus partition poly} follows from the correctness of Algorithm \ref{alg: cactus partition} and the lemmas in Subsection \ref{sec: poly time algorithm}. Algorithm \ref{alg: cactus partition poly} performs the same number of $\odot$-operations as Algorithm \ref{alg: cactus partition}, namely $\mathcal{O}(n^2)$. Because of Corollary~\ref{cor: number interval sets}, the time it takes to perform one $\odot$-operations has been reduced to $\mathcal{O}(p^4)$. We perform the merge-operation on all computed interval partition sets $I'$ and $I'_j$. The runtime of the merge-operation is linear in the number of elements in these sets. Thus, most merge-operations take $\mathcal{O}(p^4)$ time. There is one exception, which is the merge-operation that is performed on the set $I'_c$ inside \ref{alg: cyclePartition poly}. This set is the union of the interval partition sets for all configurations of a cycle and therefore contains $\mathcal{O}(np^4)$ elements. However, this specific operation is only executed once for each cycle node and thus at most $\mathcal{O}(n)$ times during the algorithm. Additionally, it takes $\mathcal{O}(p^2)$ time to check if there is a correct tuple in $I(r)$ for $r$ being the root of the tree representation and decide if a $p$-$(l,u)$-partition exists. This results in an overall runtime of $\mathcal{O}(p^4n^2)$.
	 \end{proof}
\bigskip
\begin{remark}\label{rem: bounded length/degree}
	If the length of the cycles in the given cactus graph is bounded by a constant $c$, there only $\mathcal{O}(cn)$ $\odot$-operations in Algorithm~\ref{alg: cactus partition} and \ref{alg: cactus partition poly}. Thus the runtime reduces by a factor of $n$ and equals the one for trees.
\end{remark}
Now assume that the weights on the vertices and the weight bounds $l$ and $u$ are real numbers. Because the method does not store individual weights but intervals of weights, we can compute the interval partition sets in the same way and 
that Lemma~\ref{lem: number intervals bounded} still holds in this case. Thus, Algorithm~\ref{alg: cactus partition poly} also solves the given problem for real-valued weights and weight bounds.

\subsection{Computation of partitions}\label{sec: partition computation} 
The algorithms presented in the previous subsections solved the decision variant of the $p$-$(l,u)$-partition problem. Thus, they returned either \texttt{True} or \texttt{False} depending on whether a feasible partition exists. In this section, we consider the problem of computing a feasible partition.
\begin{problem}[Computation $p$-$(l,u)$-partition problem]
	Let $G=(V,E,w)$ be a vertex-weighted graph. Given two non-negative integers $l$ and $u$ with $l\leq u$, find a $p$-$(l,u)$-partition of $G$ if it exists. 
\end{problem}
Let $G=(V,E)$ be a tree (representation) with root $r$. The presented algorithms compute partition sets such that each element in a partition set corresponds to a certain partition of the considered subgraph. Remark~\ref{rem: existence partition} and Lemma~\ref{lem: poly existence partition} state that if the partition set $S(r)$ contains a certain element there exists a feasible partition of $G$. We can compute the corresponding partition by storing additional information for each element during the computation and use backtracking afterwards. For the simple algorithms (Alg.~\ref{alg: tree partition} and \ref{alg: cactus partition}), the approach is straightforward, but the polynomial-time algorithm requires additional procedures. In the following, we present the method for trees and explain the required additions for cactus graphs.\par 
Let $P=\{V_1,\ldots,V_p\}$ be a partition of a tree $T=(V,E)$ with root $r$. We can describe the partition as the set of edges $E_P$ such that the deletion of the edges in $E_P$ from $E$ results in connected components that correspond to the clusters $V_i$. Each tuple in a partition set was computed by applying the $\odot$-operation on a certain edge in the graph. The idea is to remember if in the corresponding partition this edge is contained in a single cluster or connects two clusters. If the tuple resulted from the $\otimes$-operation, the edge is contained in a cluster. If the tuple resulted from the $\oplus$-operation, the edge connects two clusters. Using this information, we can decide during the backtracking if the edge has to be added to the set $E_P$.\par 
\paragraph{Simple computation}
First, we show the computation method for the simple algorithms. For each tuple $(x,k)\in S(v,i)$, we store a corresponding tuple $(b,x',k')$ containing the following elements: $(x',k')$ is the tuple in $S(v_i)$ with which $(x,k)$ was computed. The Boolean value $b$ is 1 if the tuple was computed with the $\oplus$-operation and 0 if the $\otimes$-operation was used. Note that it suffices to store only one such tuple for each element in the partition sets. The partition sets $S(v,0)$ are initialized with a tuple $(x,k)=(w(v),1)$ and for this element we store a \texttt{None} entry, which indicates that the backtracking can stop at this point.
During the backtracking process, we consider some tuple $(x,k)\in S(v,i)$ and its corresponding tuple $(b,x',k')$. If $b=1$, we include the edge $(v,v_i)$ in $E_P$ and continue the search with the tuples $(x,k-k')\in S(v,i-1)$ and $(x',k') \in S(v_i)$. If $b=0$, we continue with $(x-x',k-k'+1)\in S(v,i-1)$ and $(x',k')\in S(v_i)$ without adding the edge to the partition.\par 
Assume that the partition sets and tuples have been precomputed by Algorithm~\ref{alg: tree partition} and the algorithm returned \texttt{True}. There are $\mathcal{O}(n)$ partition sets of size $\mathcal{O}(up)$ in which we search for a certain element. Therefore, the backtracking requires $\mathcal{O}(upn)$ time to compute a feasible partition. Thus, the computation problem can be solved in $\mathcal{O}(u^2p^2n)$ time if Algorithm~\ref{alg: tree partition} is applied.\par 
Let $T$ be the tree representation of a cactus graph. In this case, Algorithm~\ref{alg: cactus partition} is used to compute the partition sets. For each tuple $(x,k)\in S_j(v,i)$ that was computed for a certain configuration $j$ during \ref{alg: cyclePartition}, we store a tuple $(j,b,x',k')$ for the backtracking. The additional parameter $j$ represents the configuration and and thus determines in which partition sets the computation continues. In the set $S_c$, the partition sets for all configurations are united and duplicates are removed. Note that it suffices to keep only one tuple $(j,b,x',k')$ for each element $(x,k)$. During the backtracking process when encountering a tuple $(j,b,x',k')$, we consider the $j$-th configuration of the corresponding cycle $C=\langle w_1,\ldots,w_m\rangle $. In this configuration, the edge $(w_{j},w_{j+1})$ is removed from the cycle and can be added to the set $E_P$. The vertices $w_{j}$ and $w_{j+1}$ are either in two different clusters or contained in the same one. In the first case, the edge $(w_{j},w_{j+1})$ connects two clusters and therefore has to be added to the partition. The second case implies that all vertices of the cluster are contained in the same cluster. Therefore, the connectivity of the cluster stays preserved even if all edges of $E_P$ (including $(w_{j},w_{j+1})$) are removed from the graph.
For all partition sets that are not computed with \ref{alg: cyclePartition}, we store tuples $(b,x',k')$ and perform the backtracking as above. The computation of a partition requires $\mathcal{O}(upn)$ time and thus the computation problem can be solved in $\mathcal{O}(u^2p^2n^2)$ time if Algorithm~\ref{alg: cactus partition} is applied.\par 

\paragraph{Polynomial-time computation}
To the best of our knowledge there was no method presented yet how to solve the computation problem in polynomial time. The idea of the polynomial-time partition algorithm is to reduce the size of the computed partition sets by saving merged intervals instead of individual weights. However, this means that the information about the partition that corresponds to a certain weight is lost. From Lemma~\ref{lem: poly existence partition}, we merely know that there exists a weight corresponding to a feasible partition in some interval $[a,a']$, but it is unclear how to access this specific weight and its corresponding partition. We propose a method that stores additional information about the merged intervals to provide the desired backtracking while retaining a polynomial runtime.\par 
Note that Algorithm~\ref{alg: cactus partition poly} solves the decision problem not only for cactus graphs but for trees as well. The \ref{alg: cyclePartition poly} procedure is only executed for cactus graphs. 
Let $([x,y],k)\in I(v,i)$ be a tuple in some interval partition set $I'(v,i)=I(v,i-1)\odot I(v_i)$. Similar to before, we store a corresponding tuple $(b,[x',y'],k')$. If $b=1$, the tuple was computed with the $\oplus$-operation, otherwise the $\otimes$-operation was used. For elements in some set $I_j$, which is computed when partitioning a cycle, the tuple contains an additional parameter $j$ that indicates the configuration. In the following, we describe the approach for the general sets $I$. The backtracking process for sets $I_j$ works equivalently with similar additions as presented for the simple computation approach.\par 
\begin{remark}
	Let $([x,y],k)$ be an element in $I(v,i)$ and $X$ be the maximal subset of elements in $I'(v,i)$ such that $M(X)=\{([x,y],k)\}$. An element $([x',y'],k)\in I'(v,i)$ is contained in $X$ if and only if $[x',y']\subseteq [x,y]$.
\end{remark}

The backtracking process begins at the root of the tree. Let the root node~$r$ have $j$ children and $r_j$ is the last child. 
We select a tuple $([x,y],p)\in I(r)$ such that $[x,y]\cap [l,u]\neq \emptyset$. We know that there exists a weight $a\in [x,y]$ such that $l\leq a\leq u$. The weight $a$ is also contained in some $[x',y']\cup [x,y]$ such that $([x',y'],p)\in I'(v,i)$.
Therefore, we search $I'(v,i)$ for an element $([x',y'],p)$ such that $[x',y']\cup[x,y]$ and $[x,y]\cap [l,u]\neq \emptyset$. There might be multiple tuples that satisfy this condition and we select one of them. This element $([x',y'],p)$ has a corresponding tuple $(b,[x_1,y_1],k_1)$.
The backtracking continues with $([x_1,y_1],k_1)\in I(r_j)$ and some element $([x_2,y_2],k_2)\in I(r,j-1)$ that is determined by the value of $b$. In the following steps of the backtracking, the search objective changes. If $b=0$, we have $([x',y'],p)=([x_1+x_2],k_1+k_2-1)$. In this case, we know that $[x_1+x_2,y_1+y_2]\cap [l,u]\neq \emptyset$. When searching for an element $([x'_1,y'_1],k_1)\in I'(r_j)$, it has to fulfill $[x'_1+x_2,y'_1+y_2]\cap [l,u]\neq \emptyset$ (and $[x'_1,y'_1]\subseteq [x_1,y_1]$). Additionally, we search for an element $([x'_2,y'_2],k_2)\in I'(r,j_1)$. Here it is not enough to find an element such that 
$[x_1+x'_2,y_1+y'_2]\cap [l,u]\neq \emptyset$, because it has to fulfill  $[x'_1+x'_2,y'_1+y'_2]\cap [l,u]\neq \emptyset$ as well. Note that not all tuples $[x'_2,y'_2]\in I'(v,i-1)$ with $[x'_2,y'_2]\subseteq[x_2,y_2]$ fulfill $[x'_1+x'_2,y'_1+y'_2]\cap [l,u]\neq \emptyset$, but there exists at least one tuple that fulfills this condition.
\begin{lemma}
	Let $[x_1,y_1]$ and $[x_2,y_2]$ be two $d$-consecutive intervals such that $[x_1+x_2,y_1+y_2]\cap[l,u]\neq\emptyset$. Let $X_1$ and $X_2$ be a set of $d$-interfering $d$-consecutive intervals such that $M(X_1)=[x_1,y_1]$ and $M(X_2)=[x_2,y_2]$. For every $[x'_1,y'_1]\in X_1$ fulfilling $[x'_1+x_2,y'_1+y_2]\cap[l,u]\neq \emptyset$ exists an interval $[x'_2,y'_2]$ such that $[x'_1+x'_2,y'_1+y'_2]\cap[l,u]\neq \emptyset$.
\end{lemma}
\begin{proof}
	Because the intervals are consecutive, we have $m_1$ values $a_i$ in $[x_1,y_1]$ such that $x_1=a_1 < a_2 < \ldots < a_m=y_1$ and $a_i-a_{i-1}\leq d=u-l$. Analogously, we have $m_2$ values $b_i$ in $[x_2,y_2]$ fulfilling the same property. For every $a_i$ exists some interval $[x'_1,y'_1]\in X_1$ with $a_i\in [x'_1,y'_1]$. The same holds for $b_i$ and $X_2$.
	Let $[x'_1,y'_1]$ be an interval in $X_1$. The interval $[x'_1+x_2,y'_1+y_2]$ consists of all values $a+b$ with $a\in [x'_1,y'_1]$ and $b\in[x_2,y_2]$. If $[x'_1+x_2,y'_1+y_2]\cap[l,u]\neq \emptyset$, there exists a value $x=a+b\in [x'_1+x_2,y'_1+y_2]$ with $l\leq x\leq u$. Let $[x'_2,y'_2]\in X_2$ be the interval such that $b\in [x'_2,y'_2]$. Obviously, $x$ is contained in $[x'_1+x'_2,y'_1+y'_2]$ and thus  $[x'_1+x'_2,y'_1+y'_2]\cap[l,u]\neq \emptyset$.
	 \end{proof}
\bigskip
In later stages of the backtracking a number of previous selections have to be taken into account for the search objective. Whenever we encounter a tuple with $b=1$ and an edge $(v,v_i)$ is added to the partition $E_P$, the search objective resets for the element $([x',y'],k')\in I'(v_i)$, which only has to fulfill $[x',y']\cap[l,u]\neq \emptyset$. An exemplary search query for the tree presented in Figure~\ref{fig: computation} is demonstrated in Example~\ref{exam: computation}. Note that because the computation of the partition sets was conducted in a bottom-up and left-to-right manner, the backtracking process follows a top-down and right-to-left approach.

\begin{example}\label{exam: computation}
	We illustrate the computation method for the tree $T$ shown in Figure~\ref{fig: computation tree} and present the first few steps of the backtracking process, which computes a feasible partition of $T$. We restrict this example to the computation for the edges between the nodes $r$ and $v_1$ to $v_4$. The complete computation continues in their respective subtrees. If $[a,b]\cap[l,u]\neq\emptyset$ for some interval $[a,b]$, we say that $[a,b]$ fulfills (*). Given some $([x_i,y_i],k_i)\in I(v)$, whenever we say to find an element $([x'_i,y'_i],k_i)\in I'(v)$, it is implied that $[x'_i,y'_i]\subseteq[x_i,y_i]$ should hold.\par 
	Using Algorithm~\ref{alg: cactus partition poly}, we computed the interval partition sets depicted in Figure~\ref{fig: computation backtracking}. 
	If the algorithm returned \texttt{True}, there is an element $([x,y],p)\in I(r)$ such that $[x,y]$ fulfills (*). We find $([x',y'],p)\in I'(r)$ with $(b,[x_1,y_1],k_2))$ such that $[x',y']$ fulfills (*). Let us assume that $b=0$.
	\begin{enumerate}
		\item Consider $([x_1,y_1],k_1)\in I(v_2)=I(v_2,2)$. Find $([x'_1,y'_1],k_1)\in I'(v_2)$ such that $[x'_1+x_2,y'_1+y_2]$ fulfills (*). Let us assume that $(1,[x_3,y_3],k_3)$ is the corresponding tuple. The edge $(v_2,v_4)$ is added to edge set $E_P$ of the partition.
		\item  Consider $([x_2,y_2],k_2)=([x'-x_1,y'-y_1],p-k_1+1)\in I(r,i-1)$. Find $([x'_2,y'_2],k_2)\in I'(r,i-1)$ such that $[x'_1+x'_2,y'_1+y'_2]$ fulfills (*). Let us assume that $(0,[x_5,y_5],k_5)$ is the corresponding tuple.
		\item  Consider $([x_3,y_3],k_3)\in I(v_4)$. Find $([x'_3,y'_3],k_3)\in I'(v_4)$ such that $[x'_3,y'_3]$ fulfills (*). The computation continues based on the corresponding tuple.
		\item Consider $([x_4,y_4],k_4)=([x'_1,y'_1],k_1-k_3)\in I(v_2,j-1)$. Find $([x'_4,y'_4],k_4)\in I'(v_2,j-1)$ such that $[x'_4+x'_2,y'_4+y'_2]$ fulfills (*). Let us assume that $(0,[x_7,y_7],k_7)$ is the corresponding tuple.
		\item Consider $([x_5,y_5],k_5)\in I(v_1)$. Find $([x'_5,y'_5],k_5)\in I'(v_1)$ such that $[x'_5+x_6+x'_4,y'_5+y_6+y'_4]$ fulfills (*).  The computation continues based on the corresponding tuple.
		\item Consider $([x_6,y_6],k_6)=([x'_2-x_5,y'_2-y_5],k_2-k_5+1)\in I(r,i-2)$. Find $([x'_6,y'_6],k_6)\in I'(r,0)$ such that $[x'_5+x'_6+x'_4,y'_5+y_6'+y'_4]$ fulfills (*).  The computation continues based on corresponding tuple.
		\item Consider $([x_7,y_7],k_7)\in I(v_3)$. Find $([x'_7,y'_7],k_7)\in I'(v_3)$ such that $[x_7'+x_8+x'_5+x'_6, y_7'+y_8+y'_5+y'_6]$ fulfills (*).  The computation continues based on the corresponding tuple.
		\item Consider $([x_8,y_8],k_8)=([x'_5-x_7,y'_5-y_7],k_5-k_7+1)\in I(v_2,j-2)$. Find $([x'_8,y'_8],k_8)\in I'(v_2,0)$ such that $[x'_8+x'_7+x'_5+x'_6, y'_8+y'_7+y'_5+y'_6]$ fulfills (*).  The computation continues based on the corresponding tuple.
	\end{enumerate}
\end{example}

\begin{figure}[t]
	\centering
	\begin{subfigure}{0.4\textwidth}
		\includegraphics[scale=0.9]{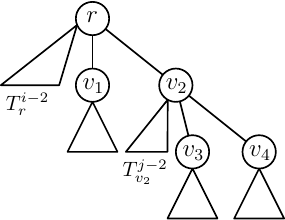}
		\caption{}
		\label{fig: computation tree}
	\end{subfigure}\hfill
	\begin{subfigure}{0.5\textwidth}
		\includegraphics[scale=1]{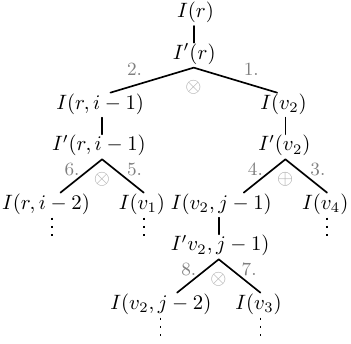}
		\caption{}
		\label{fig: computation backtracking}
	\end{subfigure}
	\caption{Example for the computation of a partition of the tree shown in (a). The backtracking process searches through the interval partition sets shown in (b) as explained in Example~\ref{exam: computation}.}
	\label{fig: computation}
\end{figure}

\begin{theorem}
	Let $G=(V,E,w)$ be a vertex-weighted graph. Given a positive integer $p$ and two non-negative integers $l$ and $u$ with $l\leq u$, we can solve the computation $p$-$(l,u)$-partition problem in $\mathcal{O}(p^4n)$ time if $G$ is a tree and $\mathcal{O}(p^4n^2)$ time if $G$ is a cactus graph.
\end{theorem}
\begin{proof}
	We use Algorithm~\ref{alg: cactus partition poly} to compute interval partition sets. On each edge in the graph, we perform either exactly one $\odot$-operation to compute an interval partition set $I'$ or we perform $\mathcal{O}(n)$ $\odot$-operations to compute sets $I'_j$. For each computed element $([x,y],k)$, we store one corresponding tuple $((j),b,[x',y'],k')$. This does not influence the time or space complexity of this operation. The resulting interval partition sets $I'_{(j)}$ have size $\mathcal{O}(p^4)$. More precisely, they contain at most $p^3$ elements for each value of $k\ (\leq p)$. As mentioned before, the runtime of the merge-operation is linear in the number of elements. The merging-process results in at most $k$ tuples for which we store a corresponding list $L$. Note that each of the previous $p^3$ elements will be added to exactly one list $L$. This happens for each value of $k$. Therefore, storing this information for one interval partition set $I$ or $I_j$ requires $\mathcal{O}(p^4)$ space. Again, there is one exception that is the computation of the set $I_c$ in \ref{alg: cyclePartition poly}, where the merging-operation requires $\mathcal{O}(p^4n)$ time and space.\par 
	After all partition sets are computed and the algorithm returned \texttt{True}, we compute a partition using the backtracking approach described in this subsection. We search for a feasible tuple in the interval partition set, which has size $\mathcal{O}(p^2)$. Then, we search for a feasible interval in its corresponding list $L$ and continue the computation. For trees, the lists have size $\mathcal{O}(p^3)$ and therefore these search procedures require $\mathcal{O}(p^2+p^3)=\mathcal{O}(p^3)$ time for each partition set and the entire backtracking takes $\mathcal{O}(p^3n)$ time. For cactus graphs, there are at most $\mathcal{O}(n)$ partition sets, where the search procedure requires $\mathcal{O}(p^3n)$ time and thus the backtracking takes $\mathcal{O}(p^3n^2)$ time. In both cases, the overall runtime is dominated by the execution of Algorithm~\ref{alg: cactus partition poly}, which requires $\mathcal{O}(p^4n)$ time for a tree and $\mathcal{O}(p^4n^2)$ for a cactus graph.
	 \end{proof}

\subsection{MinNum- and MaxNum-\texorpdfstring{$(l,u)$}{(l,u)}-partition problem}\label{sec: minimum and maximum partition}
So far, we considered $(l,u)$-partitions with a fixed number of clusters. Now, we are interested in partitions of optimal size.
\begin{problem}[MinNum/MaxNum $(l,u)$-partition problem]
	Let $G=(V,E,w)$ be a vertex-weighted graph. Given two non-negative integers $l$ and $u$ with $l\leq u$, find a $(l,u)$-partition of $G$ with the minimum or respectively maximum number of clusters.
\end{problem} 
Let $G$ be a cactus graph and $r$ the root of its tree representation. We define $k_{min}$ as number of clusters in a MinNum-$(l,u)$-partition. For a MaxNum-$(l,u)$-partition, $k_{max}$ is defined analogously. By applying Algorithm \ref{alg: cactus partition poly} with $p=n$, we can find $k_{min}$ and $k_{max}$ as follows:
\begin{align*}
k_{min} &= \min\{k \ \vert\ ([a,a'],k)\in I(r) \text{ s.t. } [a,a']\cap [l,u]\neq \emptyset\}\\
k_{max} &= \max\{k \ \vert\ ([a,a'],k)\in I(r) \text{ s.t. } [a,a']\cap [l,u]\neq \emptyset\}.
\end{align*}
Algorithm \ref{alg: cactus partition poly} takes $\mathcal{O}(n^6)$ time and searching for $k_{min}$ and $k_{max}$ takes $\mathcal{O}(n^2)$ time. 
If the graph is a tree, we can find the minimum or maximum number of cluster in $\mathcal{O}(n^5)$ time~\cite{Ito2012}. With the computation method presented in the previous subsection, we can compute the corresponding partitions and thus solve the MinNum- and MaxNum-$(l,u)$-partition problem with the same runtime.
\begin{theorem}
	Let $G=(V,E,w)$ be a vertex-weighted graph. Given two non-negative integers $l$ and $u$ with $l\leq u$, the MinNum- and MaxNum-$(l,u)$-partition problems can be solved in $\mathcal{O}(n^5)$ time if $G$ is a tree and in $\mathcal{O}(n^6)$ time if $G$ is a cactus graph.
\end{theorem}

\section{Other weight-constrained partition problems}\label{sec: other problems}
The partition methods presented in Section \ref{sec: general l,u partition problems} can be used as an algorithmic framework to solve other $(l,u)$-partition problems. Until now, we only considered vertex-weighted graphs. Note that the previous algorithms work analogously for edge-weighted graphs, where the weight of a cluster is defined as the sum of the weights on the edges inside the cluster. If the graph has both weights on its vertices and costs on its edges, further problems arise. On the one hand, one can include additional constraints and, on the other hand, one can consider different optimization goals. Many of these problems become NP-hard even on trees. Here, we show for a selection of problems how our method can be adjusted to solve these problems for both trees and cactus graphs in pseudopolynomial time.

\subsection{MinCost partition problem}\label{sec: min cost problem}
Let $G=(V,E,w,c)$ be graph with weights $w(v)$ on its vertices and costs $c(e)$ on its edges. For a partition $P$, we define the \emph{cost} $C(P)$ as the sum of the costs of all edges outside of the clusters, i.e. edges $(v,v')$ such that $v$ and $v'$ are not in the same cluster in $P$. The \emph{MinCost-$(l,u)$-partition problem} is defined as follows. See Figure~\ref{fig: example mincost} for an example.
\begin{problem}[MinCost-$(l,u)$-partition problem]
	Let $G=(V,E,w,c)$ be a graph as defined above. Given two non-negative integers $l$ and $u$ with $l\leq u$, find a $(l,u)$-partition $P$ of $G$ such that $C(P)$ is minimized.
\end{problem} 

First, let us assume that the graph is a tree. Even without a lower weight bound $l$, this problem is NP-hard, which can be shown with a reduction of the knapsack problem. For this case, Lukes presented a pseudopolynomial-time algorithm with a runtime of $\mathcal{O}(u^2n)$~\cite{lukes1974efficient} and Johnson and Niemi presented a $\mathcal{O}(un^2)$ algorithm~\cite{johnson1983knapsacks}. We show how this problem can be solved, when a lower weight bound $l$ is applied.\par 
We redefine the parameter $k$ in the partition sets as the cost of a partition. Thus, an element $(x,k)\in S(v,i)$ corresponds to an extendable $(l,u)$-partition $P$ of the subtree $T_v^i$ such that the weight of the cluster containing the node $v$ is $x$ and the cost of the partition is $k$.
\begin{align*}
S(v,i) = \{(x,k)\ \vert\ &\exists\text{ extendable $(l,u)$-partition $P$ of $T_v^i$}\\
& \text{s.t. $C(P) = k$ and $w(P_v)=x$}\}.
\end{align*}
Let $r$ be the root of the tree. The MinCost-$(l,u)$-partition corresponds to the tuple $(x,k)\in S(r)$ such that $k$ is the smallest cost over all tuples $(x,k)$ fulfilling $l\leq x\leq u$.\par

\begin{figure}[h]
	\centering
	\begin{subfigure}[c]{0.48\textwidth}
		\includegraphics[scale=0.8]{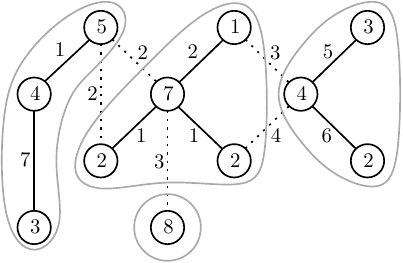}
		\subcaption{MinNum-(3,12)-partition.}
	\end{subfigure}
	\begin{subfigure}[c]{0.48\textwidth}
		\includegraphics[scale=0.8]{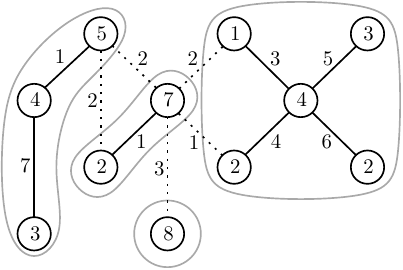}
		\subcaption{MinCost-(3,12)-partition.}
	\end{subfigure}
	\caption{Example of a MinNum- and MinCost-$(l,u)$-partition of a vertex- and edge-weighted cactus graph. The values inside the vertices are their weights. The left partition has cost 14 and the right one has cost 10, which is the minimum cost.}
	\label{fig: example mincost}
\end{figure}
We compute partition sets using the $\odot$-operation as before, but the $\oplus$- and $\otimes$-operations have to be adjusted. Whenever the two partition sets $S(v,i-1)$ and $S(v_i)$ are combined, their costs are summed up. Remember that the $\oplus$-operation considers the case that the two clusters containing the node $v$ and respectively $v_i$ are not merged. Therefore, the edge $(v,v_i)$ connects two nodes in different clusters and its cost $c(v,v_i)$ has to be added to the partition costs. 
\begin{align*}
S_1 \oplus S_2 &= \{(x_1,k_1+k_2+c(v,v_i))\ \vert\ l\leq x_2, (x_1,k_1)\in S_1, (x_2,k_2)\in S_2\}\\
S_1 \otimes S_2 &= \{(x_1+x_2,k_1+k_2)\ \vert\ x_1+x_2\leq u, (x_1,k_1)\in S_1, (x_2,k_2)\in S_2\}
\end{align*}
We only keep tuples with minimal cost. For this, we use a $\min$-operation, which reduces a partition set to tuples $(x,k_{min})$ with $k_{min}$ being the minimum cost for a specific weight $x$. 
\begin{align*}
S(v,0) &= \{(w(v),0)\}\\
S(v,i) &= \min(S(v,i-1)\odot S(v_i))
\end{align*}
Note that reducing the tuples to those with minimal cost can also be done during the computation of partition sets. For some weight $x$, we can update $k_{min}$ whenever we find a partition with smaller cost. As the size of each partition set is $\mathcal{O}(u)$, the computation takes $\mathcal{O}(u^2)$ time. Hence, we can solve the MinCost-$(l,u)$-partition problem on trees in $\mathcal{O}(u^2n)$ time.\par 
When the considered graph is a cactus graph, we can adapt Algorithm \ref{alg: cactus partition} in a similar way. More adjustments have to be made in the \ref{alg: cyclePartition} procedure. Note that in every configuration of a cycle $C=\langle w_1,\ldots,w_m \rangle $ one specific edge is removed from $C$. However, its cost does not add up to the overall cost of the partition if all vertices of $C$ are contained in one single cluster. Therefore, we include an additional Boolean value $b$ in the tuples of the partition sets $S_j$, which indicates whether a cut in the cycle occurred. For all tuples in the partition sets $S_j(v,0)$, we initialize $b=0$. Let $v$ be a vertex in $C$ and respectively $T(C,j)$. We consider the computation $S_j(v,i)=\min(S_j(v,i-1)\odot S_j(v_i))$.
Let $(x_1,k_1,b_1)$ and $(x_2,k_2,b_2)$ are two tuples that are combined during the computation of a partition set $S_j$. The $\oplus$-operation does not merge the two clusters, a cut occurs and the resulting tuple is $(x_1,k_1+k_2+c(v,v_i),1)$. The $\otimes$-operation checks whether a cut occurred before and the resulting tuple is $(x_1+x_2,k_1+k_2, b_1\lor b_2)$. Let $w_1$ be the root of all tress $T(C,j)$. After computing $S_j(w_1)$, we asses the values $b$ and add the cost of the missing edge if a cut occurred. Thus, we compute $\widehat{S}_j(w_1) = \{(x,k+k')\ \vert \ (x,k,b)\in S_j(w_1)\}$ where $k'=c(w_{j},w_{j+1})$ if $b=1$ and $k'=0$ otherwise. Then, we compute $S_c$ as the union of all $\widehat{S}_j(w_1)$ and again reduce this set to only tuples with minimal cost. This partition set is then returned by \ref{alg: cyclePartition}. The runtime of the algorithm is still dominated by the time it takes to perform the $\odot$-operation.
\begin{theorem}
	Let $G=(V,E,w,c)$ be a graph and $l$ and $u$ two integers as defined above. The min-cost $(l,u)$-partition problem can be solved in $\mathcal{O}(u^2n)$ time if $G$ is a tree and in $\mathcal{O}(u^2n^2)$ time if $G$ is a cactus graph.
\end{theorem}
It is noteworthy that in the MinNum-$(l,u)$-partition problem, the addition of a lower weight constraint increased the runtime for trees significantly. Namely, from a linear runtime~\cite{kundu1977linear} to $\mathcal{O}(n^5)$~\cite{Ito2012}. In the min-cost $(l,u)$-partition problem, we were able to obtain the same runtime with a lower weight bound than Lukes~\cite{lukes1974efficient} did without one. Moreover, this method can also be used to solve the \emph{MinCost-p-(l,u)-partition problem}, which is finding a $(l,u)$-partition consisting of $p$ clusters with minimal cost. In this case, we include the size as well as the cost in the tuples of the partition sets and obtain a solution in $\mathcal{O}(u^2p^2n)$ time for trees and  $\mathcal{O}(u^2p^2n^2)$ time for cactus graphs.

\subsection{MinMax and MaxMin partition problems}\label{sec: minmax and maxmin problem}
Let $G$ be a graph in which each vertex is not only assigned a weight $w(v)$ but also a size $s(v)$. Similarly to $w(V_i)$, the \emph{size} $s(V_i)$ of a vertex cluster $v_i$ is now defined as the sum of the sizes of its vertices. The \emph{MinMax-$p$-$(l,u)$-partition problem} is defined as follows. See Figure~\ref{fig: example minmax} for an example.
\begin{problem}[MinMax-$p$-$(l,u)$-partition problem]
	Let $G=(V,E,w,s)$ be a graph as defined above. Given two non-negative integers $l$ and $u$ with $l\leq u$ and a positive integer $p\leq n$, find a $p$-$(l,u)$-partition $P$ of $G$ such that the size of the largest cluster is minimized. 
\end{problem} 
\begin{figure}[b]
	\centering
	\begin{subfigure}[c]{0.48\textwidth}
		\includegraphics[scale=0.8]{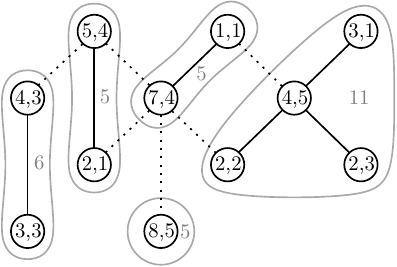}
		\subcaption{5-(3,12)-partition.}
		\label{fig: minmax 1}
	\end{subfigure}
	\begin{subfigure}[c]{0.48\textwidth}
		\includegraphics[scale=0.8]{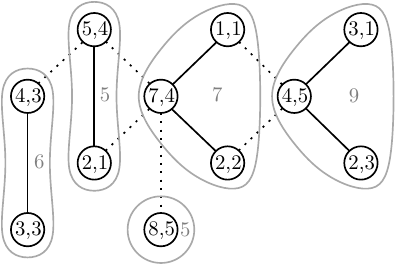}
		\subcaption{MinMax-5-(3,12)-partition.}
		\label{fig: minmax 2}
	\end{subfigure}
	\caption{Example of a MinMax-$p$-$(l,u)$-partition of a cactus graph. The value $x,y$ inside the vertices are their weight ($x$) and size ($y$). The gray numbers are the sizes of the clusters. Solving the $p$-$(l,u)$-partition problem for $p=5$, $l=3$ and $u=12$ with Algorithm~\ref{alg: cactus partition poly} might result in the partition shown in (a). With the adjusted algorithm, we can compute a 5-(3,12)-partition in which the weight of the largest cluster is minimized. This partition is shown in (b). In this case, this is also a possible MaxMin partition.}
	\label{fig: example minmax}
\end{figure}
In the related \emph{MaxMin} problem the size of the smallest cluster is maximized.
The standard MaxMin-$p$-partition problem, which does not include any weight constraints, can be solved in linear time on trees~\cite{frederickson1991optimal}. Agasi et al. proved that the addition of a upper weight constraint $u$ results in NP-hardness and presented a pseudopolynomial-time algorithm with a runtime of $\mathcal{O}(u^2n^3\log(s(V)))$, where $s(V)=\sum_{v\in V}s(v)$ is the size of the entire tree~\cite{agasi1993shifting}. We show that even with an additional lower weight bound, we obtain an algorithm with a similar runtime. \par 
First, we consider the decision variant of the given problem on trees: Given a tree $T$, non-negative integers $l$, $u$ ($l\leq u$) and $u_s$ and a positive integer $p\leq n$, is there partition $P=\{V_1,\ldots, V_p\}$ of $G$ with $p$ clusters such that $l\leq w(V_i)\leq u$ and $s(V_i)\leq u_s$ for all $V_i$? Again, we assume that for each vertex $v$ we have $w(v)\leq u$ and $s(v)\leq u_s$. Similarly to Section~\ref{sec: min cost problem}, we redefine the partition sets and adjust their computation accordingly. Every element $(x,y,k)\in S(v,i)$ corresponds to an extendable $(l,u)$-partition $P$ of $T_v^i$ with $\vert P\vert =k$, $w(P_v)=x$ and $s(P_v)=y$.
\begin{align*}
S_1 \oplus S_2 = &\{(x_1,y_1, k_1+k_2)\ \vert\ l\leq x_2, k_1+k_2\leq p,(x_1,y_1,k_1)\in S_1, \\
& (x_2,y_2,k_2)\in S_2\}\\
S_1 \otimes S_2=  &\{(x_1+x_2,y_1+y_2, k_1+k_2)\ \vert\ x_1+x_2\leq u, y_1+y_2\leq u_s,\\ &k_1+k_2-1\leq p, (x_1,y_1,k_1)\in S_1, (x_2,y_2,k_2)\in S_2\}
\end{align*}
After each $\odot$-operation and for each combination of $x$ and $k$, we keep only one tuple $(x,y_{min},k)$ where $y_{min}$ is minimal. Therefore, the number of elements in a partition set remains $\mathcal{O}(up)$ and the computation for the entire tree needs $\mathcal{O}(u^2p^2n)$ time. By using binary search over the value $u_s$, we can solve the min-max $p$-$(l,u)$-partition problem in time $\mathcal{O}(u^2p^2n\log(s(V)))$.\par 
To solve the corresponding MaxMin problem, we consider the decision problem with a parameter $l_s$ and the constraint $l\leq s(V_i)$ for all clusters $V_i$ in the partition. In this case, we keep only tuples $(x,y_{min},k)$ and perform binary search over the value $l_s$. This results in the same overall runtime.  
If the graph $G$ is a cactus graph, we can adjust the algorithm in a similar way and solve both problems with an additional factor of $n$ in the runtime. 
\begin{theorem}
	Let $G=(V,E,w,s)$ be a graph and $l$, $u$ and $p$ integers as defined above. The MinMax- and MaxMin-$(l,u)$-partition problems can be solved in $\mathcal{O}(u^2p^2n\log(s(V)))$ time if $G$ is a tree and in $\mathcal{O}(u^2p^2n^2\log(s(V)))$ time if $G$ is a cactus graph.
\end{theorem}
\paragraph{Special case}
In the case that sizes equal weights, i.e. $s(v)=w(v)$, we can also consider the following problem: Find a $p$-$(l,u)$-partition such that the weight of each cluster is bounded by $l$ and $u$, but the weight of the heaviest cluster has to be minimized (or the weight of the lightest cluster maximized). Then, we have two possible approaches with pseudopolynomial runtime. The first approach is to binary search over the parameter $u$ to find the best value. For each considered value of $u$, we apply the polynomial-time algorithm for the 
$p$-$(l,u)$-partition problem as shown in Section \ref{sec: poly time algorithm}. For trees, this approach has a runtime of $\mathcal{O}(p^4n\log(u))$. The second approach is to compute partition sets consisting of tuples $(x,y,k)\in S(v,i)$ corresponding to a partition $P$ of $T_v^i$ with size $k$ such that $x=w(P_v)$ and $y$ is the weight of the heaviest cluster in~$P$. 
\begin{align*}
S_1 \oplus S_2 = &\{(x_1,\max(y_1,y_2), k_1+k_2)\ \vert\ l\leq x_2,\ k_1+k_2\leq p,\\
&(x_1,y_1,k_1)\in S_1,\ (x_2,y_2,k_2)\in S_2\}\\
S_1 \otimes S_2=  &\{(x_1+x_2,\max(y_1,y_2,x_1+x_2), k_1+k_2)\ \vert\ x_1+x_2\leq u,\\ 
& k_1+k_2-1\leq p, (x_1,y_1,k_1)\in S_1,\ (x_2,y_2,k_2)\in S_2\}
\end{align*}
Again, we keep only tuples $(x,y_{min},k)$ in our computation. Then, the desired partition corresponds to the tuple in which $y_{min}$ is the minimum value over all tuples $(x,y,p)\in S(r)$ fulfilling $l\leq x\leq u$. For trees, this approach has a runtime of $\mathcal{O}(u^2p^2n)$, which might be favorable for smaller values of~$u$.

\subsection{Capacity constrained \texorpdfstring{$(l,u)$}{(l,u)}-partition problem}\label{sec: capacity constrained problem}
Let $G=(V,E,w,c)$ be a graph with weights $w(v)$ on its vertices and capacities $c(e)$ on its edges. Given a vertex partition $P$ of $G$, the \emph{capacity} of a cluster $V_i$ is defined as the sum of the capacities of all edges that connect a vertex from $V_i$ to a vertex from another cluster, i.e. $c(V_i)=\sum_{v\in V_i, v'\notin V_i}c(v,v')$. The \emph{MinNum-$(l,u,u_c)$-partition problem} is defined as follows. See Figure~\ref{fig: example capacity} for an example.
\begin{problem}[MinNum-$(l,u,u_c)$-partition problem]
	Let $G=(V,E,w,c)$ be a graph as defined above. Given three non-negative integers $l$, $u$ and $u_c$ with $l\leq u$, find a $(l,u)$-partition $P$ of $G$ such that $c(V_i)\leq u_c$ for all clusters $V_i\in P$ and the number of clusters is minimized.
\end{problem} 
\begin{figure}[b]
	\centering
	\begin{subfigure}[c]{0.48\textwidth}
		\includegraphics[scale=0.8]{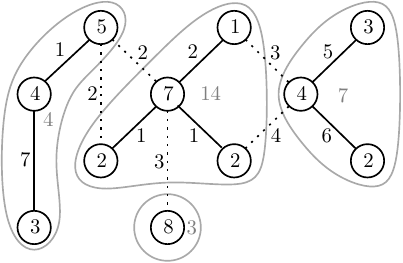}
		\subcaption{MinNum-(3,12)-partition.}
		\label{fig: capacity 1}
	\end{subfigure}
	\begin{subfigure}[c]{0.48\textwidth}
		\includegraphics[scale=0.8]{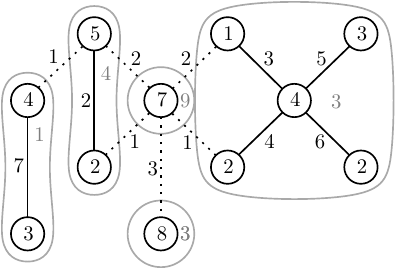}
		\subcaption{MinNum-(3,12,9)-partition.}
		\label{fig: capacity 2}
	\end{subfigure}
	\caption{Example of a MinNum-$(l,u)$- and $(l,u,u_c)$-partition of a cactus graph. The values inside the vertices are their weights and the gray numbers are the capacities of the clusters. Solving the MinNum-$(l,u)$-partition problem for $l=3$ and $u=12$ with Algorithm~\ref{alg: cactus partition poly} might result in the partition shown in (a). If the capacity of the clusters is bounded by $u_c=9$, we can find a feasible partition as in (b) with the adjusted algorithm. Note that the additional constraint can increase the size of the partition.}
	\label{fig: example capacity}
\end{figure}
Perl and Snir considered this problem without the lower weight bound $l$ and showed that it is NP-hard even if the graph is a tree. They presented pseudopolynomial algorithms with different runtimes, namely  $\mathcal{O}(u^2n^3)$, $\mathcal{O}(u_c^2n^3)$ and $\mathcal{O}(u_cn^4)$~\cite{perl1983circuit}.\par
First, we consider a tree $T$ with root $r$. Again, we can redefine the partition sets to include capacities as well. A tuple $(x,y,k)\in S(v,i)$ corresponds to an extendable $(l,u)$-partition $P$ of $T_v^i$ with $k$ clusters such that $x=w(P_v)$ and $y=c(P_v)$ is the capacity of the cluster containing the node $v$. Obviously, only edges in the subtree $T_v^i$ are considered for the capacity. We adjust the computation accordingly and compute
\begin{align*}
S(v,0) &= \{(w(v),0,0)\}\\
S(v,i)&=\min(S(v,i-1)\odot S(v_i)).
\end{align*}

In the $\oplus$-operation, the two considered clusters are not merged and therefore the capacity for both clusters increases by the capacity of the edge $(v,v_i)$. Hence, we have to check if the resulting capacities fulfill the capacity constraint $u_c$. For $S_1=S(v,i-1)$ and $S_2=S(v_i)$, we define: 
\begin{align*}
S_1 \oplus S_2 = &\{(x_1,y_1+c(v,v_i),k_1+k_2)\ \vert\ l\leq x_2,\ y_2+c(v,v_i)\leq u_c,\ \\ 
&y_2+c(v,v_i)\leq u_c,\ (x_1,y_1,k_1)\in S_1,\ (x_2,y_2,k_2)\in S_2\}\\
S_1 \otimes S_2 = &\{(x_1+x_2,y_1+y_2,k_1+k_2-1)\ \vert\ x_1+x_2\leq u,\ y_1+y_2\leq u_c\\ 
&(x_1,y_1,k_1)\in S_1,\ (x_2,y_2,k_2)\in S_2\}
\end{align*}
Note that it is not necessary to keep all computed tuples. Since the capacity for the clusters is only bounded from above, it suffices to keep for each combination of $x$ and $k$ the one with the smallest capacity.
We use the $\min$-operation to keep only tuples $(x,y_{min},k)$. Thus, the number of elements in the partition sets remains $\mathcal{O}(un)$ and the computation requires $\mathcal{O}(u^2n^2)$ time. In the end, the MinNum-$(l,u,u_c)$-partition corresponds to the tuple in $S(r)$ in which $k$ is the minimum over all tuples $(x,y,k)\in S(r)$ satisfying $l\leq x\leq u$. Using this approach, we can solve the problem in time $\mathcal{O}(u^2n^3)$ for trees.\par 
In case the graph $G$ is a cactus graph rooted in a hinge $r$, some additional adjustments are needed to partition a cycle. Let $C=\langle w_1,w_2,\ldots,w_m\rangle $ be a cycle in $G$. In some configuration $j$ of $C$, the edge $(w_{j},w_{j+1})$ is removed for the following computation. If the capacity of this edge is not added to the capacities of the clusters containing $w_{j}$ and $w_{j+1}$ respectively, then both vertices have to end up in the same cluster - which has to contain all vertices of the cycle. We compute two different partition sets. The sets $S''_j$ are computed with the $\odot$-operation and the sets $S'_j$ are obtained only with $\otimes$-operations. 
\begin{align*}
S''_j(v,i) &= \min(S''(v,i-1)\odot S''_j(v_{i}))\\
S'_j(v,i) &= \min(S'(v,i-1)\otimes S'_j(v_{i}))
\end{align*}
We initialize $S'_j(w_i,0)$ as in Section \ref{sec: simple algorithm} and additional sets $S''_j(w_i,0)$ as follows:
\begin{align*}
S''_j(w_i,0) = \begin{cases}
\{(x,y+c(w_{j},w_{j+1}),k)\ \vert\ (x,y,k)\in S'_j(w_i,0)\} &\mbox{if } i=j \text{ or } j+1,\\
S'_j(w_i,0) &\mbox{otherwise.}
\end{cases}
\end{align*}
In the end, we compute the set $S_c$, which is returned by the procedure, as follows:
\begin{equation*}
S_c = \min\left(\bigcup_{j=1}^{m-1} (S'_j(w_1)\cup S''_j(w_1))\right).
\end{equation*}
Note that for every tuple $(x,y,k)\in S'_j(w_1)$, we have a corresponding tuple $(x,y+2 c(w_{m-j+1},w_{m-j+2}),k)\in S''_j(w_1)$, which resulted from elements obtained by repeatedly merging clusters, i.e. only $\otimes$-operations inside the $\odot$-operations. 
The latter does not contain the correct capacity because all vertices of the cycle are contained in the same cluster. However, this incorrect tuple is removed with the $\min$-operation, as it has a higher capacity value than $(x,y,k)$. The computation of two partition sets for each edge inside a cycle does not change the overall runtime of this algorithm.
\begin{theorem}
	Let $G=(V,E,w,c)$ be a graph and $l$, $u$ and $u_c$ three integers as defined above. The MinNum-$(l,u,u_c)$-partition problem can be solved in $\mathcal{O}(u^2n^3)$ time if $G$ is a tree and in $\mathcal{O}(u^2n^4)$ time if $G$ is a cactus graph.
\end{theorem}

\section{Conclusion and further work}\label{sec: conclusion}
We considered different $(l,u)$-partition problems for weighted trees and cactus graphs. We presented a method to partition a cycle inside a cactus graph efficiently and proved that the resulting algorithm solves the general $(l,u)$-partition problems in polynomial time. For the partition into a fixed number $p$ of clusters, we obtained an algorithm with a runtime of $\mathcal{O}(p^4n^2)$. We showed that the minimum or maximum number of clusters can be found in $\mathcal{O}(n^6)$ time. Whereas previous research focused on solving the decision problems, we considered the computation problem as well and presented a method that computes the desired partitions in polynomial time.\par 
As we briefly explained, the presented partition algorithm for cactus graphs can be used to obtain a skeleton-based decomposition for simple polygons, which allows us to solve the corresponding decomposition problem in polynomial time. If the degree of the branching points in the skeleton is bounded by some constant, which is the case in the discrete skeleton that we apply, the decomposition problem can be solved in $\mathcal{O}(n^5)$ time (see Remark~\ref{rem: bounded length/degree}).\par
Furthermore, we showed that the presented partition approach can be used as an algorithmic framework to solve other partition problems for weighted trees and cactus graphs. We considered a selection of different $(l,u)$-partition problems that include additional constraints. All these problems known to be NP-hard even for trees. We showed how our partition algorithm can be adjusted to solve these problems in pseudopolynomial time for both trees and cactus graphs. All previous algorithmic results considered only trees and $l=0$. For other problems such as the MinNum-$(l,u)$-partition problem, we can observe that the runtime of the algorithms increased considerably with the addition of a lower weight bound $l\leq 0$. With our adjusted partition algorithm, we were able to include a lower weight bound without an increase in runtime.\par 
Our partition approach can potentially be adjusted to solve further partition problems as well. The general $p$-$(l,u)$-partition problem proved to be solvable in polynomial time for trees and cactus graphs but NP-hard for graphs with bounded treewidth. The complexity for other graph classes such as for example outerplanar graphs remains open.

\bibliographystyle{abbrv}
\bibliography{references}  

\end{document}